\newenvironment{proof}{{\indent \indent \it Proof:\quad}}{\hfill $\blacksquare$\par}
\documentclass[preprint,12pt]{elsarticle}
\usepackage{}
\usepackage{floatrow}
\usepackage{algorithm}
\usepackage{algorithmic}
\newfloatcommand{figurebox}{figure}[\nocapbeside][\dimexpr(\textwidth-\columnsep)/2\relax]
\newfloatcommand{tablebox}{table}[\nocapbeside][\dimexpr(\textwidth-\columnsep)/2\relax]
\usepackage{booktabs}
\usepackage{threeparttable}
\usepackage{amscd}
\usepackage{amssymb}
\usepackage{amsmath}
\usepackage{amsfonts}
\usepackage[pagewise]{lineno}
\usepackage{float}
\usepackage{algorithm,algorithmic}

\usepackage{eqparbox}

\usepackage{epstopdf}
\usepackage{float}  %
\usepackage{subfigure}  %
\usepackage{graphicx}
\usepackage{epstopdf}
\usepackage{algorithm,algorithmic}
\usepackage{caption}
\usepackage{verbatim}
\usepackage{array}
 \usepackage{tabularx}
\usepackage{multirow}
 \usepackage{multicol}
 \usepackage{arydshln}
\newtheorem{lemma}{Lemma}
\newtheorem{theorem}{Theorem}
\newtheorem{definition}{Definition}

\usepackage{threeparttable}
\def\max{\mbox{max}\,}

\usepackage{amssymb}

\journal{Nuclear Physics B}

\begin{document}

\begin{frontmatter}



\title{ Low Rank Quaternion Matrix Recovery via Logarithmic Approximation}
\author[lab1]{Liqiao Yang}
\ead{liqiaoyoung@163.com}
\author[lab1]{Jifei Miao}
\ead{jifmiao@163.com}
\author[lab1]{Kit Ian Kou\corref{cor1}}
\ead{kikou@umac.mo}

\address[lab1]{Department of Mathematics, Faculty of
	Science and Technology, University of Macau, Macau 999078, China}

\cortext[cor1]{Corresponding author}

\begin{abstract}
In color image processing, image completion aims to restore missing entries from the incomplete observation image. Recently,  great progress has been made in achieving completion by approximately solving the rank minimization problem. In this paper, we utilize a novel quaternion matrix logarithmic norm to approximate  rank under the quaternion matrix framework. From one side, unlike the traditional matrix completion method that handles RGB channels separately, the quaternion-based method is able to avoid destroying the structure of images via putting the color image in a pure quaternion matrix. From the other side, the logarithmic norm induces a more accurate rank surrogate. Based on the logarithmic norm, we take advantage of not only truncated technique but also factorization strategy to achieve image restoration. Both strategies are optimized based on the  alternating minimization framework. The experimental results demonstrate that the use of logarithmic surrogates in the quaternion domain is more superior in solving the problem of color images completion.      	 	
\end{abstract}




\end{frontmatter}


\section{Introduction}
\label{sec:intro}
Image completion is one of the key branches in image processing, which has attracted extensive research interests in recent years \cite{DBLP:journals/pami/HuZYLH13,DBLP:journals/pami/ShangCLLL18,DBLP:journals/tip/ChenXZ20,DBLP:journals/tsp/MiaoK20}. When processing color images, the most traditional matrix-based image completion algorithms handle the three-dimensional data by unfolding it into
matrices, which breaks the inherent correlation among RGB channels. It would cause some important information to be lost. This method of matrix-based image completion  is devoted to solving the following model
\begin{equation}\label{m1}
\min\limits_{\mathbf{X}} \mathit{f}(\mathbf{X}),  \qquad  \text{ s.t.}   \quad   P_\Omega(\mathbf{X}-\mathbf{M})=0,
\end{equation}
where $\mathit{f}(\mathbf{X})$ is the cost function, $\mathbf{M}\in \mathbb{R}^{m\times n}$ is the observed matrix with missing entries, and $\mathbf{X}\in \mathbb{R}^{m\times n}$ is the recovered matrix. The linear operation $P_\Omega(\cdot)$ is developed to indicate the elements in $\Omega$ with 1 representing observed changeless elements and 0 representing missing elements. Since the structure of visual data has strong non-local self-similarity, it normally exhibits low rank characteristics \cite{DBLP:journals/tip/ZhouLLZ18}. Therefore, the cost function in (\ref{m1})  can be replaced by the rank function, and the low rank matrix completion (LRMC) model is formulated as
\begin{equation}\label{m2}
\min\limits_{\mathbf{X}} \text{rank} (\mathbf{X}),  \qquad  \text{ s.t.}   \quad   P_\Omega(\mathbf{X}-\mathbf{M})=0.
\end{equation}
As a combined optimization problem, problem (\ref{m2}) is generally solved by optimizing a convex proxy of the rank function \cite{DBLP:journals/focm/CandesR09}, because directly dealing with the rank function  is NP-hard. In the previous study \cite{DBLP:journals/focm/CandesR09}, the nuclear norm of the matrix ($\parallel \mathbf{X}\parallel_*=\sum_{i=1}^{\min(m,n)}\sigma_{i}(\mathbf{X})$, where $\sigma_{i}(\mathbf{X})$ is the i\emph{th}  singular value of  $\mathbf{X}$) is proved to be effective in replacing the rank function.

Many existing studies demonstrate that LRMC algorithms based on the nuclear norm may produce sub-optimal results because every singular value is processed equally \cite{DBLP:journals/pami/HuZYLH13, DBLP:conf/cvpr/GuZZF14}. This approach  goes against the fact that a larger singular value may contain more rich and useful information than a smaller one \cite{DBLP:journals/pami/HuZYLH13,DBLP:journals/tip/XieGLZZZ16}. To overcome these  shortcomings, many nonconvex alternatives have been designed. For example, the truncated nuclear norm $\parallel\mathbf{X}\parallel_r=\sum_{i=r+1}^{\min(m,n)}\sigma_{i}(\mathbf{X})$ (TNN),  the sum of $\min(m,n)-r$ minimum singular values of $\mathbf{X}$, is proposed in \cite{DBLP:journals/pami/HuZYLH13}. Moreover, the proposed two-step optimization process in \cite{DBLP:journals/pami/HuZYLH13} is effective.There are other similar surrogates to TNN, such as weighted nuclear norm \cite{DBLP:conf/cvpr/GuZZF14} and weighted Schatten p-norm \cite{DBLP:journals/tip/XieGLZZZ16}. Similarly, the log-determinant alternate, \emph{ i.e.}, $L(\mathbf{X},\epsilon)=\sum_{i=r+1}^{\min(m,n)}\log(\sigma_{i}(\mathbf{X})+\epsilon)$ with $\epsilon>0$, is also proved to be a more precise approximation of the rank than nuclear norm \cite{DBLP:journals/cin/KangPCC15,xie2018tensor}. The value of the logarithmic function becomes gentle as the singular values increase, meanwhile, the smaller singular values can be penalized more.  

The strategies mentioned above are based on optimizing the approximation of the rank, which need to process whole singular value decomposition (SVD) of the matrix $\mathbf{X}$. Whereas, the computational complexity of SVD is high, which makes the promotion of high dimensional or big data have limitations \cite{DBLP:journals/pami/ShangCLLL18, DBLP:journals/mpc/WenYZ12, DBLP:journals/corr/LinCM10}.  To address this concern, the low rank matrix factorization (LRMF) is applied to matrix completion. These methods describe the low rank property of target matrix $\mathbf{X}$ by factorizing it to two smaller factor matrices, \emph{ i.e.},  $\mathbf{X}=\mathbf{U}\mathbf{V}^T$ , where $\mathbf{U}\in \mathbb{R}^{m\times r}$ and $\mathbf{V}\in \mathbb{R}^{n\times r}$ with $r\ll \min {(m,n)}$. Although this method naturally satisfies the low rank requirement of the target matrix and also benefits from fast numerical optimization, it lacks the estimated value rank of  $\mathbf{X}$ in many situations. In \cite{DBLP:journals/mpc/WenYZ12}, this problem is solved by an approximate skill. Besides, using a factorized approach instead of the spectral norm defined by one kind of weighted nuclear norm, would lead to a local minimum solution \cite{DBLP:journals/jmlr/AbernethyBEV09}.   

There is a necessary dimensional reduction operation  based on LRMC to process color image completion. Because the above matrix-based approaches are focused on two dimensional data, color images can only be processed separately for RGB channels.  In this way, they may lose the inner relationships among the three channels. In this regard,  color image processing based on the quaternion framework has received widespread attention. Thanks to the structure of  quaternion itself, each pixel of the color image can be represented by a pure quaternion to form a quaternion matrix. This representation is fully utilized in color images edge detection  \cite{hu2018phase}, color face recognition  \cite{DBLP:journals/access/ZouKDZT19}, color image denoising \cite{DBLP:journals/ijon/YuZY19,DBLP:journals/mssp/GaiYW015}, color image completion \cite{DBLP:journals/tip/ChenXZ20,DBLP:journals/tsp/MiaoK20}, and so on. 

In particular, when processing  color image completion in the quaternion domain, the authors  extended a low rank quaternion approximation (LRMA) model from LRMC \cite{DBLP:journals/tip/ChenXZ20}. It based on  three nonconvex  functions for approximating the rank of the quaternion matrix, including weighted Schatten p-norm, Laplace approximation, and Geman approximation.  All these functions highlight the advantages of the quaternion matrix that have been verified experimentally and theoretically. However, as described above in the matrix cases,  these approaches need to process fully  quaternion singular value decomposition (QSVD) for quaternion matrix with  expensive computational cost. To overcome this  imperfection, the LRMF strategies are extended to the quaternion domain  \cite{DBLP:journals/tsp/MiaoK20}. The authors factorized the target quaternion matrix to two-factor quaternion matrices and designed three quaternion-based bilinear factor matrix norm factorization methods  for low-rank quaternion matrix completion (LRQMC). The methods of LRQMC include quaternion double Frobenius norm (Q-DFN), quaternion double nuclear norm (Q-DNN), and quaternion Frobenius/nuclear norm (Q-FNN).  All variables are optimized under the  alternating direction method of multipliers (ADMM) framework which is similar to the real field. These factorized-based approaches only need to handle two smaller quaternion matrices. Thus considerable  computational cost could be reduced.
 
In order to approximate the rank more accurately and effectively in the quaternion domain that takes full advantage of the color image structure, we develop two algorithms for LRQMC named Quaternion Logarithmic norm-based Factorization (QLNF)  and Truncated Quaternion Logarithmic norm-based Approximation (TQLNA). Meanwhile, these two algorithms can narrow the gap between low-rank minimization and low-rank factorization.  In this paper, a new quaternion logarithmic norm (QLN) is introduced as a non-convex replacement of rank, where the low rank  can be more accurately depicted than traditional approximation methods like quaternion nuclear norm (QNN). Then, we adopt QLN to  two smaller quaternion matrices,  factor quaternion matrices of the target quaternion matrix in QLNF. In TQLNA, we truncate QLN and then optimize the target quaternion matrix directly by using the truncated quaternion  logarithmic norm (TQLN). Concretely, the main contributions of this paper are summarized as follows: 
\begin{itemize}
	\item We propose a new nonconvex surrogate QLN for the rank of quaternion matrix with the guarantee of a surrogate theorem, and  act QLN on factor quaternion matrix. 
	
	\item  The rank of the quaternion matrix would not be changed by the influence of the largest $r$ singular values. So the proposed TQLN is operated by truncating the largest $r$ singular values and then by adopting QLN in the truncated issue.
	
	\item For surrogate QLN, the fast iterative shrinkage thresholding algorithm (FISTA)  is used to optimize the factor quaternion matrix. For TQLN, the ADMM based algorithm is used to efficiently solve the optimization problems. The experimental results prove that our method is effective in color image completion. 
\end{itemize}

The rest of this paper is organized in the following structure. In Section \ref{sec:1}, related notations and preliminary knowledge in the quaternion domain will be reviewed. In Section \ref{sec:2}, we will discuss the quaternion-based model for image completion and depict QLNF and TQLNA algorithms. In Section \ref{Experimental results}, this study will display the numerical experiment to compare other algorithms, and then a conclusion will be presented in Section \ref{Conclusion}. Part of the proofs will be listed in the Appendix.

\section{Notations and preliminaries}\label{sec:1}

In this section, the main notations and  the preliminary knowledge based on quaternions adopted throughout this paper are introduced briefly.

\subsection{Notations}
we use a, $\textbf{a}$, and  $\textbf{A}$ to  denote scalar, vector and matrix in real domain $\mathbb{R}$, respectively. A dot above the variables ($\dot{a}$, $\dot{\textbf{a}}$, and $\dot{\textbf{A}}$) are to denote the variables  ( scalars, vectors, and matrices) in the quaternion domain $\mathbb{H}$. Besides, the complex space is denoted by $\mathbb{C}$. Both $\textbf{I}_{r\times r}$ and $\textbf{I}_{r}$ represent the $r\times r$ identity matrix.  $\mathfrak{R}(\dot{a})$ is the real part of quaternion $\dot{a}$.  $\mathbf{(\cdot)}^T$, $\mathbf{(\cdot)}^*$, $\mathbf{(\cdot)}^H$ and $\mathbf{(\cdot)}^{-1}$ represent the transpose, conjugate, conjugate transpose and the inverse of  $\mathbf{\cdot}$, respectively. $|\cdot|$, $\|\cdot\|_F$, $\|\cdot\|_*$ are the absolute value of modulus, Frobenius norm, and nuclear norm of $\cdot$, respectively. The inner product of $\circ_1$ and $\circ_2$ is defined as $\langle\circ_1\cdot\circ_2\rangle\triangleq \text{tr}(\circ_1^H\circ_2)$, and tr$(\cdot)$ is the trace function.

\subsection{Preliminary sketch of quaternion-based knowledge}
Quaternion space $\mathbb{H}$ was first proposed by W. Hamilton  \cite{doi:10.1080/14786444408644923}.  As a generalization of complex number, a quaternion number $\dot{a}\in\mathbb{H}$ can be represented by Cartesian form 
\begin{equation}\label{m3}
 \dot{a}=a_0+a_1\emph{i}+a_2\emph{j}+a_3\emph{k},
\end{equation}
  where $a_l\in\mathbb{R} (l=0,1,2,3)$, and \emph{i, j, k} are three imaginary number units which have the following relationships
  \begin{equation}\label{m4}
  \begin{cases}
  \emph{i}^2= \emph{j}^2 =\emph{k}^2= \emph{i}\emph{j}\emph{k}=-1\\
  \emph{i}\emph{j}=-\emph{j}\emph{i} = \emph{k},    \emph{j}\emph{k}=-\emph{k}\emph{j} = \emph{i},  \emph{k}\emph{i}=-\emph{i}\emph{k} = \emph{j}.
  \end{cases}
  \end{equation}
    $\mathfrak{R}(\dot{a}) \triangleq a_0$ is the real part of  $\dot{a}$. $\mathfrak{I}(\dot{a}) \triangleq a_1\emph{i}+a_2\emph{j}+a_3\emph{k}$ is the  imaginary part of $\dot{a}$.  Hence $\dot{a}=\mathfrak{R}(\dot{a})+\mathfrak{I}(\dot{a})$. Besides, when real part $a_0 = 0$, $\dot{a}$ is  a pure quaternion.  The conjugate and the modulus of $\dot{a}$  are defined as: $\dot{a}^{*} = a_0-a_1\emph{i}-a_2\emph{j}-a_3\emph{k}$  and $|\dot{a}|= \sqrt{\dot{a}\dot{a}^{*}}=\sqrt{a_0^2+a_1^2+a_2^2+a_3^2}$, separately. Assuming two quaternions $\dot{a}$ and $\dot{b}\in\mathbb{H}$,  the addition and multiplication are  separately defined as following 

\begin{equation}\nonumber 
   \dot{a}+\dot{b}=(a_0+b_0)+(a_1+b_1)\emph{i}+(a_2+b_2)\emph{j}+(a_3+b_3)\emph{k}
\end{equation}
\begin{equation}\nonumber
 \begin{aligned}
   \dot{a}\dot{b}=& (a_0b_0-a_1b_1-a_2b_2-a_3b_3)\\&+(a_0b_1+a_1b_0+a_2b_3-a_3b_2)\emph{i}
   \\&+(a_0b_2-a_1b_3+a_2b_0+a_3b_1)\emph{j}\\&+(a_0b_3+a_1b_2-a_2b_1+a_3b_0)\emph{k}.
 \end{aligned}
\end{equation}
What calls for special attention is that the multiplication in the quaternion domain is not commutative $\dot{a}\dot{b} \neq \dot{b}\dot{a}$.

The quaternion matrix $\dot{\textbf{A}}=(\dot{a}_{ij})\in\mathbb{H}^{M \times N}$, where $\dot{\textbf{A}}=\textbf{A}_0+\textbf{A}_1\emph{i}+\textbf{A}_2\emph{j}+\textbf{A}_3\emph{k}$. In the definition of  $\dot{\textbf{A}}$,  $\textbf{A}_l\in\mathbb{R}^{M \times N} (l=0,1,2,3)$ which are real matrices. When  $\textbf{A}_0 = \textbf{0}$, $\dot{\textbf{A}}$ is a pure quaternion matrix. The Forbenius norm is defined as:  $\parallel\dot{\textbf{A}}\parallel_F = \sqrt{\sum_{i=1}^{M}\sum_{j=1}^{N}|\dot{a}_{ij}|^2}=\sqrt{tr(\dot{\textbf{A}}^H\dot{\textbf{A}})}$. The form using Cayley-Dickson notation \cite{DBLP:journals/sigpro/BihanM04} is $\dot{\textbf{A}}=\textbf{A}_p+\textbf{A}_q\emph{j}$, where $\textbf{A}_p$ and $\textbf{A}_q\in\mathbb{C}^{M \times N}$ are complex matrices. Hence, the representation of  quaternion matrix $\dot{\textbf{A}}$ can be denoted as an isomorphic complex  matrix $\textbf{A}_c\in\mathbb{C}^{2M \times 2N}$
\begin{equation}\label{ct}
\textbf{A}_c={
\left( \begin{array}{cc}
\textbf{A}_p & \textbf{A}_q  \\
-\textbf{A}_q^* & \textbf{A}_p^*\\
\end{array}
\right )_{2M \times 2N}},
\end{equation}
where  $\textbf{A}_p=\textbf{A}_0+\textbf{A}_1\emph{i}\in\mathbb{C}^{M \times N}$ and $\textbf{A}_q=\textbf{A}_2+\textbf{A}_3\emph{i}\in\mathbb{C}^{M \times N}$.

\begin{definition}(\textbf{The rank of quaternion matrix} \cite{zhang1997quaternions}): For a quaternion matrix $\dot{\textbf{A}}=(\dot{a}_{ij})\in\mathbb{H}^{M \times N}$, the  maximum number of right (left) linearly independent columns (rows)  is defined as the rank of  $\dot{\textbf{A}}$.
\end{definition}

\begin{theorem}(\textbf{QSVD} \cite{zhang1997quaternions}):
\label{theorem1}
 Given a quaternion matrix $\dot{\textbf{A}}\in\mathbb{H}^{M \times N}$ be of rank
$r$. There are two unitary quaternion matrices $\dot{\textbf{U}}\in\mathbb{H}^{M \times M}$
and $\dot{\textbf{V}}\in\mathbb{H}^{N \times N}$ such that
\begin{equation}
\dot{\mathbf{A}}=\dot{\mathbf{U}}
\left( \begin{array}{cc}
\mathbf{\Sigma}_r & \mathbf{0}  \\
\mathbf{0} & \mathbf{0}\\
\end{array}
\right )\dot{\mathbf{V}}^H= \dot{\mathbf{U}}\mathbf{\Lambda}\dot{\mathbf{V}}^H,
\end{equation}
where $\mathbf{\Sigma}_r=diag({\sigma_1,\cdots, \sigma_r})\in\mathbb{R}^{r\times r}$, and all singular values $\sigma_i>0,  i=1,\cdots,r$ . 
\end{theorem}

The connection of the QSVD  of  $\dot{\textbf{A}}$  and the SVD of the isomorphic complex  matrix $\mathbf{A}_c\in\mathbb{C}^{2M \times 2N}$ ($\mathbf{A}_c=\mathbf{U}\acute{\mathbf{\Lambda}}\mathbf{V}^H$) is listed as follows \cite{DBLP:journals/tip/XuYXZN15}
\begin{equation}
\begin{cases}
\mathbf{\Lambda}=row_{odd}(col_{odd}(\acute{\mathbf{\Lambda}})),\\
\dot{\mathbf{U}}=col_{odd}(\mathbf{U}_1)+col_{odd}(-(\mathbf{U}_2)^*)\emph{j},\\
\dot{\mathbf{V}}=col_{odd}(\mathbf{V}_1)+col_{odd}(-(\mathbf{V}_2)^*)\emph{j},
\end{cases}
\end{equation}
then we have $\dot{\mathbf{A}}= \dot{\mathbf{U}}\mathbf{\Lambda}\dot{\mathbf{V}}^H$, where
$\mathbf{U}=\left( \begin{array}{c}
\mathbf{U}_1)_{M\times2M}  \\
(\mathbf{U}_2)_{M\times2M}\\
\end{array}
\right )$ and 
$\mathbf{V}=\left( \begin{array}{c}
\mathbf{V}_1)_{N\times2N}  \\
(\mathbf{V}_2)_{N\times2N}\\
\end{array}
\right )$.  The odd rows
and odd columns of $\star$ are represented by $row_{odd}(\star)$ and $col_{odd}(\star)$ respectively.

Utilizing the above property and theorem, the QSVD can be obtained by computing the classical SVD of the complex matrix $\textbf{A}_c$. Moreover, we can see that the number of positive singular values $r$ is the rank of  $\dot{\mathbf{A}}$ \cite{zhang1997quaternions}.  For a more detailed introduction of quaternion algebra, please refer to \cite{zhang1997quaternions, girard2007quaternions}. Following the Theorem \ref{theorem1}, we can get the definition of quaternion matrix nuclear norm (QNN). 

\begin{definition}(\textbf{QNN}) \cite{DBLP:journals/tip/ChenXZ20,DBLP:journals/ijon/YuZY19}:
	 	\label{def2}
	Given $\dot{\textbf{A}}\in\mathbb{H}^{M \times N}$, the nuclear norm of the quaternion matrix is $\parallel\dot{\textbf{A}}\parallel_*=\sum_{i=1}^{min(M,N)}\sigma_{i}(\dot{\textbf{A}})$, where $\sigma_{i}$ can be obtained by the QSVD of  $ \dot{\textbf{A}}$.
\end{definition}

\begin{theorem}(\textbf{Binary factorization framework} \cite{DBLP:journals/tsp/MiaoK20}):
	\label{theorem2}
	For any quaternion matrix $\dot{\textbf{X}}\in\mathbb{H}^{M \times N}$ with $rank(\dot{\textbf{X}})=r \leqslant d$. Then the binary factorization framework is devised as
	\begin{equation}
	\dot{\textbf{X}}=\dot{\textbf{U}}\dot{\textbf{V}}^H,
	\end{equation}
	where $\dot{\textbf{U}}\in\mathbb{H}^{M \times d}$
	and $\dot{\textbf{V}}\in\mathbb{H}^{N \times d}$ such that $rank(\dot{\textbf{U}})=rank(\dot{\textbf{V}})=r.$
\end{theorem}

\begin{definition}(\textbf{Q-DFN} \cite{DBLP:journals/tsp/MiaoK20}):
	\label{def3}	
Based on theorem \ref{theorem2}, the quaternion double Frobenius norm is defined as	
	\begin{equation}
\parallel\dot{\textbf{X}}\parallel_{Q-DFN}:=\min\limits_{\begin{array}{c}
	\dot{\textbf{U}}, \dot{\textbf{V}} \\ \dot{\textbf{X}}=\dot{\textbf{U}}\dot{\textbf{V}}^H\end{array}} \frac{1}{2}\parallel\dot{\textbf{U}}\parallel_F^2+\frac{1}{2}\parallel\dot{\textbf{V}}\parallel_F^2.
\end{equation}	
\end{definition}
Essentially, whether it is based on matrix or quaternion matrix, the double Frobenius norm is the nuclear norm \cite{DBLP:journals/pami/ShangCLLL18, DBLP:journals/tsp/MiaoK20, DBLP:conf/iccv/CabralTCB13}, \textit{i.e.}, $\parallel\dot{\textbf{X}}\parallel_{Q-DFN}=\parallel\dot{\mathbf{X}}\parallel_{*}.$

\section{Quaternion-based model for image completion}\label{sec:2}
In this section, we first review the quaternion-based completion related works, and then  introduce our quaternion-based low rank completion strategies.

\subsection{Related work and low rank quaternion matrix completion }\label{sec:2.1}
In the related optimized field of completion,  a new surrogate, logarithmic norm was devised \cite{DBLP:journals/tip/ChenJLZ21}, and was utilized in factorized framework to get the recovery.  The  logarithmic norm of the given matrix $\mathbf{X}\in\mathbb{R}^{M \times N}$ is defined as
\begin{equation}
\parallel\textbf{X}\parallel_L^p=\sum_{i=1}^{min(M,N)}\log (\sigma_{i}^P(\textbf{X})+\epsilon),
\end{equation}
where $\sigma_{i}$ can be obtained by computing SVD of  $ \textbf{X}$. The main advantage of using the logarithmic norm is that it has a superior sparse substitution of rank over the convex envelope of the nuclear norm \cite{DBLP:journals/tip/ChenJLZ21, fazel2003log}. For example, when $\textbf{X}=x\in\mathbb{R}$, $rank(x)= 0$ if  $x = 0$ and $rank(x)= 1$, otherwise. Supposing $\mid x \mid \leq M$, $\frac{\parallel x \parallel_*}{M} = \frac{\mid x \mid }{M}$ is the convex envelop of $rank(x)= 0$ on $\{x\lvert \mid x \mid \leq M\}$ \cite{fazel2003log, fazel2001rank}. The intuitive explanation behind the idea of using logarithmic function to approximate rank can be found in Figure  \ref{p1}. When $\epsilon$ is a small positive constant,  logarithmic norm would be closer to $rank(x)$ than the convex envelop \cite{DBLP:journals/tip/ChenJLZ21}.
\begin{figure}
	\centering
	\includegraphics[width=80mm]{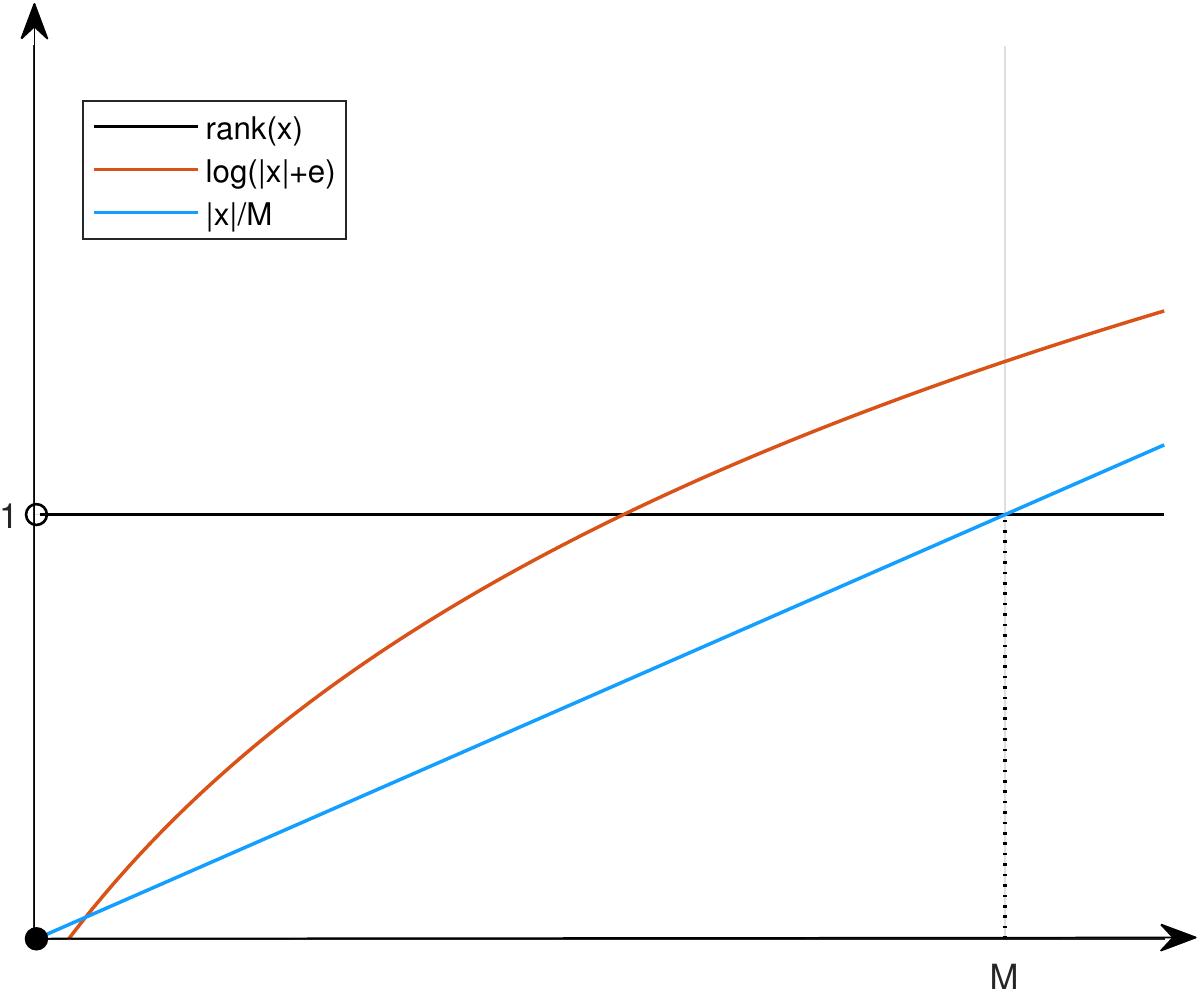}
	\caption{The rank, the convex envelop of rank (nuclear norm), and the logarithmic norm for scalar x}
	\label{p1}
\end{figure}

The classic LRMC model that  fundamentally optimizes two-dimensional data like grayscale images completion is briefly introduced in (\ref{m2}). Therefore, when processing color images,  the model in (\ref{m2}) needs to disassemble RGB channels, while LRQMC model can assemble RGB channels. As an extended model from LRMC, the model of LRQMC can be represented as \cite{DBLP:journals/tip/ChenXZ20}
\begin{equation}
\label{m5}
\min\limits_{\dot{\mathbf{X}}} \text{rank} (\dot{\mathbf{X}}),  \qquad  \text{ s.t.}   \quad   P_\Omega(\dot{\mathbf{X}}-\dot{\mathbf{M}})=0,
\end{equation}
where $\dot{\mathbf{M}}$ is the given partial observed quaternion matrix (missing color image).  $\dot{\mathbf{X}}\in \mathbb{H}^{M\times N}$ is the recovery quaternion  matrix (recovered color image). The linear operation $P_\Omega(\cdot)$  indicating the elements in $\Omega$ are presented as 1 for the observed changeless entries and 0 for missing entries.

The main completion models in the quaternion domain can  be divided into two lines: low rank minimization and low rank factorization.  Similar to LRMC, the rank function is hard to be optimized in model (\ref{m5}). Hence, for the first line, based on Definition \ref{def2},  the model of low rank minimization can be written as
\begin{equation}\label{m6}
\min\limits_{\dot{\mathbf{X}}} \parallel\dot{\mathbf{X}}\parallel_*,  \qquad  \text{ s.t.}   \quad   P_\Omega(\dot{\mathbf{X}}-\dot{\mathbf{M}})=0.
\end{equation}
Almost all optimizations  based on (\ref{m6}) are aiming to improve QNN, like a more general Schattern p-norm ($\parallel\dot{\mathbf{X}}\parallel_{Q-S_p}=(\sum_{i=1}^{min(M,N)}\sigma_{i}^{p}(\dot{\textbf{X}}))^{\frac{1}{p}}, 0<p<\infty$) \cite{DBLP:journals/tip/ChenXZ20, DBLP:journals/ijon/YuZY19}, and a Log-determinant norm ($L(\dot{\mathbf{X}},\epsilon)=\sum_{i=r+1}^{\min(m,n)}\log(\sigma_{i}(\dot{\mathbf{X}})+\epsilon)$ with $\epsilon>0$ ) \cite{DBLP:journals/tip/ChenXZ20}. When we let $p=1$, the Schattern p-norm is the frequently used nuclear norm. The mentioned methods  have proved their effectiveness theoretically and experimentally in processing color image completion. 

In addition to the methods listed above, the low rank factorization has been also found to be an important line for completion. Based on the Theorem \ref{theorem2} and Definition \ref{def3}, the model of low rank factorization can be written as \cite{DBLP:journals/tsp/MiaoK20}
\begin{equation}\label{m7}
\min\limits_{\dot{\textbf{X}}}\parallel\dot{\textbf{X}}\parallel_{Q-DFN},  \qquad  \text{ s.t.}   \quad   P_\Omega(\dot{\mathbf{X}}-\dot{\mathbf{M}})=0.
\end{equation}
After inducing a nonnegative parameter $\lambda$ to balance  the loss function and the low rank regularization, (\ref{m7}) can be further represented as
\begin{equation}
\begin{aligned}
&\min\limits_{
	\dot{\textbf{U}}, \dot{\textbf{V}} ,\dot{\textbf{X}}} \frac{\lambda}{2}(\parallel\dot{\textbf{U}}\parallel_F^2+\parallel\dot{\textbf{V}}\parallel_F^2)+\frac{1}{2}\parallel P_\Omega(\dot{\mathbf{X}}-\dot{\mathbf{M}})\parallel_F \\&
	s.t. \dot{\textbf{X}}=\dot{\textbf{U}}\dot{\textbf{V}}^H.
\end{aligned}
\end{equation}

In this way, the goal is shifted from calculating the entire QSVD of  large-scale $\dot{\textbf{X}}$ to optimizing bi-factor quaternion matrices that have smaller dimensions.  

\subsection{The proposed  Quaternion-based  Logarithmic Norm Factorization (QLNF) }
 As described above, because of the geometric structure of the quaternion itself, the processing of color images in the quaternion domain can maintain the connection between channels as much as possible. In this section, enlightened by the definition of the logarithmic norm in \cite{DBLP:journals/tip/ChenJLZ21}, we define the quaternion-based logarithmic norm. Then, we establish the  equivalence relation between quaternion-based Logarithmic norm and bi-factor quaternion matrix factorization. Eventually, the optimization process is given.
 \begin{definition}(\textbf{QLN}):
 	\label{def4}
 	Given $\dot{\textbf{X}}\in\mathbb{H}^{M \times N}$, the logarithmic norm of the quaternion matrix with $0\leq p \leq1$ and $\epsilon>0$ is
 	\begin{equation}
 \parallel\dot{\textbf{X}}\parallel_L^p=\sum_{i=1}^{min(M,N)}\log(\sigma_{i}^p(\dot{\mathbf{X}})+\epsilon),
 	\end{equation}
 where $\sigma_{i}$  can be obtained by the QSVD of  $ \dot{\textbf{X}}$.
 \end{definition}

The rank is a real number,  which is analogous to the situation in the real domain. Learning from Figure \ref{p1}, the QLN  is a more precise approximation than QNN. Moreover, inspired by the bi-factor surrogate theorem for matrix logarithmic norm in \cite{DBLP:journals/tip/ChenJLZ21}, and Theorem \ref{theorem2}, we establish the bi-factor surrogate theorem for QLN.

\begin{theorem} :
	\label{theorem3}
	For any quaternion matrix $\dot{\textbf{X}}\in\mathbb{H}^{M \times N}$ with $rank(\dot{\textbf{X}})=r \leqslant d \leq\{M, N\}$. There exist $\dot{\textbf{U}}\in\mathbb{H}^{M \times d}$
	and $\dot{\textbf{V}}\in\mathbb{H}^{N \times d}$ such that 
$\dot{\textbf{X}}=\dot{\textbf{U}}\dot{\textbf{V}}^H$. Then we have
	\begin{equation}
\parallel\dot{\textbf{X}}\parallel_{L}^{{1}/{2}}:=\mathop{\min_{
		\dot{\textbf{U}}, \dot{\textbf{V}}      \atop \dot{\textbf{X}}=\dot{\textbf{U}}\dot{\textbf{V}}^H}} \frac{1}{2}\parallel\dot{\textbf{U}}\parallel_L^1+\frac{1}{2}\parallel\dot{\textbf{V}}\parallel_L^1.
\end{equation}	
\end{theorem}

The proof of Theorem  \ref{theorem3} can be found in the Appendix.

Based on Definition \ref{def4} and Theorem  \ref{theorem3}, we present the following completion model with the quaternion-based framework
\begin{equation}\label{m8}
\min\limits_{\dot{\textbf{X}}}\parallel\dot{\textbf{X}}\parallel_{L}^{1/2},  \qquad  \text{ s.t.}   \quad   P_\Omega(\dot{\mathbf{X}}-\dot{\mathbf{M}})=0,
\end{equation}
Integrating the above questions into one formula,  (\ref{m8}) can be further written as
\begin{equation}
\min\limits_{
	\dot{\textbf{U}}, \dot{\textbf{V}} ,\dot{\textbf{X}}} \frac{\lambda}{2}(\parallel\dot{\textbf{U}}\parallel_L^1+\parallel\dot{\textbf{V}}\parallel_L^1)+\parallel \mathbf{W}\odot (\dot{\textbf{U}}\dot{\textbf{V}}^H-\dot{\mathbf{M}})\parallel_F^2 ,
\end{equation}
where $\mathbf{W} \in \mathbb{R}^{M \times N}$  denotes the location of the existing and missing entries, indicating by 1 and 0. $\odot$ is the Hadamard product. When the multiplication is acted on the real matrix $\mathbf{W}$  and the quaternion matrix $(\dot{\textbf{U}}\dot{\textbf{V}}^H-\dot{\mathbf{M}})\in \mathbb{Q}^{M \times N})$, $\odot$ still means element-wise multiplication. $\lambda$ is a nonnegative parameter.

As in matrix-based situation,  this model is  full of challenges to be optimized. Inspired by the FISTA designed in the real domain \cite{DBLP:journals/jscic/XuY17}, we induce the FISTA to update under the alternating minimization framework. It means that at each iteration, we  update factors in turn with other factors fixed.  Assuming $\dot{\textbf{U}}_{k+1}$ and $\dot{\textbf{V}}_{k+1}$   is the updating results in the k\textit{th} iteration. Specific steps are listed as follows 

\textbf{Updating $\dot{\textbf{U}}$ }:
\begin{equation}
\dot{\textbf{U}}_{k+1}= \mathop{\arg\min}\limits_{
	\dot{\textbf{U}}} \frac{\lambda}{2}\parallel\dot{\textbf{U}}\parallel_L^1+\parallel \mathbf{W}\odot (\dot{\textbf{U}}\dot{\textbf{V}}_k^H-\dot{\mathbf{M}})\parallel_F ^2.
\end{equation}
Let  $\mathcal{Q}(\dot{\textbf{U}}) = \parallel \mathbf{W}\odot (\dot{\textbf{U}}\dot{\textbf{V}}_k^H-\dot{\mathbf{M}})\parallel_F ^2$. Analogous to the ADMM framework used in the complex and the quaternion domain \cite{DBLP:journals/tsp/MiaoK20, li2015alternating} and based on FISTA,  the updating of  $\dot{\textbf{U}}_{k+1}$ can be written as:\\
\begin{equation}
\begin{cases}
 \hat{\dot{\textbf{U}}}_{k}=\dot{\textbf{U}}_{k}+\omega_k(\dot{\textbf{U}}_{k}-\dot{\textbf{U}}_{k-1})\\
 \dot{\textbf{U}}_{k+1}=\mathop{\arg\min}\limits_{\dot{\textbf{U}}} \mathfrak{R}(<\nabla\mathcal{Q}(\hat{\dot{\textbf{U}}}_{k}),  \dot{\textbf{U}} - \hat{\dot{\textbf{U}}}_{k}> )+ \frac{\mu_k}{2}\parallel\dot{\textbf{U}} - \hat{\dot{\textbf{U}}}_{k}\parallel_F^2+
 \frac{\lambda}{2}\parallel\dot{\textbf{U}}\parallel_L^1,
 \end{cases}
\end{equation}
where $t_{k}=\dfrac{1+\sqrt{1+4t_{k-1}^2}}{2}$ and $\omega_{k}=\frac{t_{k-1}-1}{t_{k}}$. Further, we can get 
\begin{equation}\label{m9}
\begin{cases}
\hat{\dot{\textbf{U}}}_{k}=\dot{\textbf{U}}_{k}+\omega_k(\dot{\textbf{U}}_{k}-\dot{\textbf{U}}_{k-1}).\\
\dot{\textbf{U}}_{k+1}=\mathop{\arg\min}\limits_{\dot{\textbf{U}}} \frac{\lambda}{2}\parallel\dot{\textbf{U}}\parallel_L^1+ \frac{\mu_k}{2}\parallel \dot{\textbf{U}} - \hat{\dot{\textbf{U}}}_{k}+ \frac{1}{\mu_k}\nabla\mathcal{Q}(\hat{\dot{\textbf{U}}}_{k})\parallel_F^2
\end{cases}
\end{equation}
It should be noted that the left and right generalized HR (GHR) of  the real functions $\mathcal{Q}(\dot{\textbf{U}})$ with quaternion variables are the same \cite{DBLP:journals/tsp/XuM15}. Consequently, according to the derivation theories for quaternion matrix in \cite{DBLP:journals/tsp/XuM15}, the gradient of  $\mathcal{Q}(\dot{\textbf{U}})$ is computed as following
\begin{equation}
\begin{aligned}
\frac{\partial\mathcal{Q}(\dot{\textbf{U}})}{\partial \dot{\textbf{U}}^*}&=
\frac{\partial Tr\{(\mathbf{W}\odot (\dot{\textbf{U}}\dot{\textbf{V}}_k^H-\dot{\mathbf{M}}))^H(\mathbf{W}\odot (\dot{\textbf{U}}\dot{\textbf{V}}_k^H-\dot{\mathbf{M}}))\}}{\partial \dot{\textbf{U}}^*}\\&
=\mathbf{W}\odot(\mathfrak{R}(\dot{\textbf{U}}\dot{\textbf{V}}_k^H)\dot{\textbf{V}}_k-\frac{1}{2}(\dot{\textbf{V}}_k\dot{\textbf{U}}^{*H})^T\dot{\textbf{V}}_k-\mathfrak{R}(\dot{\mathbf{M}})\dot{\textbf{V}}_k+\frac{1}{2}\dot{\mathbf{M}}^*\dot{\textbf{V}}_k)\\&
=\mathbf{W}\odot(\frac{1}{2}\dot{\textbf{U}}\dot{\textbf{V}}_k^H\dot{\textbf{V}}_k - \frac{1}{2}\dot{\mathbf{M}}\dot{\textbf{V}}_k).
\end{aligned}
\end{equation}
Then, (\ref{m9}) can be rewritten as
\begin{equation}\label{m10}
\begin{cases}
\hat{\dot{\textbf{U}}}_{k}=\dot{\textbf{U}}_{k}+\omega_k(\dot{\textbf{U}}_{k}-\dot{\textbf{U}}_{k-1}).\\
\dot{\textbf{U}}_{k+1}=\mathop{\arg\min}\limits_{\dot{\textbf{U}}} \frac{\lambda}{2}\parallel\dot{\textbf{U}}\parallel_L^1+ \frac{\mu_k}{2}\parallel \dot{\textbf{U}} - \hat{\dot{\textbf{U}}}_{k}+ \frac{1}{2\mu_k}\mathbf{W}\odot(\hat{\dot{\textbf{U}}}_{k}\dot{\textbf{V}}_k^H - \dot{\mathbf{M}})\dot{\textbf{V}}_k)\parallel_F^2
\end{cases}
\end{equation}
The next problem is how to obtain the optimal of  (\ref{m10}). To slove this problem, we establish the following theorem to get an approximate expression.
\begin{theorem} \textbf{(Quaternion logarithmic singular value thresholding	(QLSVT))} :
	\label{theorem4}
	For any quaternion matrix $\lambda>0$ and $\dot{\textbf{Y}}\in\mathbb{H}^{M \times N}$, the QSVD of  $\dot{\textbf{Y}}$ is $\dot{\textbf{Y}}=\dot{\mathbf{U}}_{\dot{\textbf{Y}}}(\mathbf{\Lambda}_{\dot{\textbf{Y}}})\dot{\mathbf{V}}_{\dot{\textbf{Y}}}^H$. Then the closed solution of the problem
		\begin{equation}
		\mathop{\arg\min}\limits_{\dot{\textbf{X}}} \frac{1}{2}\parallel \dot{\textbf{Y}} -\dot{\textbf{X}}\parallel_F^2+\lambda\parallel\dot{\textbf{X}}\parallel_L^1 
		\end{equation}	
		is provided by $\dot{\textbf{X}}=\dot{\mathbf{U}}_{\dot{\textbf{Y}}}\mathcal{L}_{\lambda, \epsilon}(\mathbf{\Lambda}_{\dot{\textbf{Y}}})\dot{\mathbf{V}}_{\dot{\textbf{Y}}}^H$.
	Beacuse $\mathbf{\Lambda}_{\dot{\textbf{Y}}}$  consists of real numbers, the form of  $\mathcal{L}_{\lambda, \epsilon}(\cdot)$ is similar to the real case \cite{DBLP:journals/tip/ChenJLZ21}. The soft thresholding operator $\mathcal{L}_{\lambda, \epsilon}(\cdot)$ is defined as
	\begin{equation}\label{tm1}
	\mathcal{L}_{\lambda, \epsilon}(x) : =
	\begin{cases}
	0,\qquad \qquad \qquad  \qquad  \qquad \Delta \leq 0\\
\mathop{\arg\min}\limits_{a\in \{0, \frac{1}{2}(x-\epsilon+\sqrt{\Delta})\}} \mathit{h}(a), \qquad\Delta > 0
	\end{cases}
	\end{equation}	
where $\Delta = (x-\epsilon)^2-4(\lambda-x\epsilon)$ and function $\mathit{h}(a) :=\frac{1}{2}(a-x)^2 + \lambda\log(a+\epsilon)$ is  $\mathbb{R}^+\longrightarrow\mathbb{R}^+$.
\end{theorem}
 \begin{proof}
 For quaternion matrices $\dot{\textbf{X}}\in\mathbb{H}^{M \times N}$ and $\dot{\textbf{Y}}\in\mathbb{H}^{M \times N}$, the QSVD of them is  $\dot{\textbf{X}}=\dot{\mathbf{U}}_{\dot{\textbf{X}}}(\mathbf{\Lambda}_{\dot{\textbf{X}}})\dot{\mathbf{V}}_{\dot{\textbf{X}}}^H$ and $\dot{\textbf{Y}}=\dot{\mathbf{U}}_{\dot{\textbf{Y}}}(\mathbf{\Lambda}_{\dot{\textbf{Y}}})\dot{\mathbf{V}}_{\dot{\textbf{Y}}}^H$. Then we have
\begin{equation}
\begin{aligned}
&\frac{1}{2}\parallel \dot{\textbf{Y}} -\dot{\textbf{X}}\parallel_F^2+\lambda\parallel\dot{\textbf{X}}\parallel_L^1 \\&=
\frac{1}{2}(Tr(\dot{\textbf{Y}}^H\dot{\textbf{Y}})-2Tr(\dot{\textbf{Y}}^H\dot{\textbf{X}})+
Tr(\dot{\textbf{X}}^H\dot{\textbf{X}}))+\lambda\parallel\dot{\textbf{X}}\parallel_L^1 \\&=
\frac{1}{2}(\sum\limits_{i}\sigma_i (\dot{\textbf{Y}}) -2Tr(\dot{\textbf{Y}}^H\dot{\textbf{X}})+
\sum\limits_{i}\sigma_i (\dot{\textbf{X}}))+\lambda\parallel\dot{\textbf{X}}\parallel_L^1 \\&\geq
\frac{1}{2}(\sum\limits_{i}\sigma_i (\dot{\textbf{Y}}) -2\sigma_i (\dot{\textbf{Y}}\sigma_i (\dot{\textbf{X}})+
\sum\limits_{i}\sigma_i (\dot{\textbf{X}}))+\lambda\parallel\dot{\textbf{X}}\parallel_L^1 \\&=
\sum\limits_{i}\frac{1}{2}(\sigma_i (\dot{\textbf{Y}})-\sigma_i (\dot{\textbf{X}}))^2+\lambda\log(\sigma_i (\dot{\textbf{X}})+\epsilon),
\end{aligned}
\end{equation}
where the first inequality follows from \cite{mirsky1975trace}. Based on the von Neumann's trace inequality in \cite{mirsky1975trace},  $Tr(\dot{\textbf{Y}}^H\dot{\textbf{X}})$ reaches its upper bound $\sum\limits_{i}\sigma_i (\dot{\textbf{Y}})\sigma_i (\dot{\textbf{X}})$ when $\dot{\mathbf{U}}_{\dot{\textbf{X}}}=\dot{\mathbf{U}}_{\dot{\textbf{Y}}}$, $\dot{\mathbf{V}}_{\dot{\textbf{X}}}=\dot{\mathbf{V}}_{\dot{\textbf{Y}}}$. Then we have
\begin{equation}
\begin{aligned}
&\min\limits_{\dot{\mathbf{X}}}\frac{1}{2}\parallel \dot{\textbf{Y}} -\dot{\textbf{X}}\parallel_F^2+\lambda\parallel\dot{\textbf{X}}\parallel_L^1 \Leftrightarrow \\&
\min\limits_{\{\sigma_i (\dot{\mathbf{X}})\geqslant 0\}}\sum\limits_{i}\frac{1}{2}(\sigma_i (\dot{\textbf{Y}})-\sigma_i (\dot{\textbf{X}}))^2+\lambda\log(\sigma_i (\dot{\textbf{X}})+\epsilon),\\&
 \text{s.t.}\quad \sigma_1(\dot{\textbf{X}})\geqslant\sigma_2(\dot{\textbf{X}})\geqslant\cdots\geqslant\sigma_{\min\{M,N\}}(\dot{\textbf{X}}) \geqslant 0.
\end{aligned}
\end{equation}
Next, we study how to find the lower bound of optimal $\dot{\textbf{X}}$. The singular value $\sigma_i (\dot{\textbf{X}})$ can be determined by
\begin{equation}\label{m11}
\min\limits_{\{\sigma_i (\dot{\mathbf{X}})\geqslant 0\}}\mathit{h}(\sigma_i (\dot{\textbf{X}}))=\frac{1}{2}(\sigma_i (\dot{\textbf{Y}})-\sigma_i (\dot{\textbf{X}}))^2+\lambda\log(\sigma_i (\dot{\textbf{X}})+\epsilon),
\end{equation}
which is differentiable in $[0, +\infty)$. Hence, we can get the first-order derivative of $\mathit{h}(\sigma_i (\dot{\textbf{X}}))$
\begin{equation}
\mathit{h}^{'}(\sigma_i (\dot{\textbf{X}}))=\sigma_i (\dot{\textbf{X}})-\sigma_i (\dot{\textbf{Y}})+\dfrac{\lambda}{\sigma_i (\dot{\textbf{X}})+\epsilon},
\end{equation}
 further, we have
 \begin{equation}\label{m12}
(\sigma_i (\dot{\textbf{X}})+\epsilon) \mathit{h}^{'}(\sigma_i (\dot{\textbf{X}}))=\sigma_i^2 (\dot{\textbf{X}})-(\sigma_i (\dot{\textbf{Y}})-\epsilon)\sigma_i (\dot{\textbf{X}})+\lambda-\sigma_i (\dot{\textbf{Y}})\epsilon,
 \end{equation}
 where the $\Delta_i=(\sigma_i (\dot{\textbf{Y}})-\epsilon)^2-4(\lambda-\sigma_i (\dot{\textbf{Y}})\epsilon)$. For (\ref{m12}), if $\Delta_i\leq0$, the LHS $(\sigma_i (\dot{\textbf{X}})+\epsilon) \mathit{h}^{'}(\sigma_i (\dot{\textbf{X}}))\geqslant0$. It means that $\mathit{h}(\sigma_i (\dot{\textbf{X}}))$ is monotonically increasing, so 0 is the optimal solution. If  $\Delta_i>0$, $\mathit{h}^{'}(\sigma_i (\dot{\textbf{X}})=0$  in two roots. Based on the possible monotony of the function $\mathit{h}(\sigma_i (\dot{\textbf{X}}))$, the results only include 0 root and the larger root $\frac{1}{2}(\sigma_i (\dot{\textbf{Y}})-\epsilon+\sqrt{\Delta-i})$, which is shown in (\ref{tm1}). \end{proof}

Based on the above discussion, the QLSVT can be applied to (\ref{m10}) for updating $\dot{\textbf{U}}$. By similar approach, we can get the update of $\dot{\textbf{V}}$. 

\textbf{Updating $\dot{\textbf{V}}$ }:
\begin{equation}
\dot{\textbf{V}}_{k+1}= \mathop{\arg\min}\limits_{
	\dot{\textbf{V}}} \frac{\lambda}{2}\parallel\dot{\textbf{V}}\parallel_L^1+\parallel \mathbf{W}\odot (\dot{\textbf{U}}_{k+1}\dot{\textbf{V}}^H-\dot{\mathbf{M}})\parallel_F ^2.
\end{equation}
Let  $\mathcal{P}(\dot{\textbf{V}}) = \parallel \mathbf{W}\odot (\dot{\textbf{U}}_{k+1}\dot{\textbf{V}}^H-\dot{\mathbf{M}})\parallel_F ^2$.
Based on the FISTA, the update of $\dot{\textbf{V}}_{k+1}$ can be written as 
\begin{equation}\label{m13}
\begin{cases}
\hat{\dot{\textbf{V}}}_{k}=\dot{\textbf{V}}_{k}+\omega_k(\dot{\textbf{V}}_{k}-\dot{\textbf{V}}_{k-1}).\\
\dot{\textbf{V}}_{k+1}=\mathop{\arg\min}\limits_{\dot{\textbf{V}}} \frac{\lambda}{2}\parallel\dot{\textbf{V}}\parallel_L^1+ \frac{\mu_k}{2}\parallel \dot{\textbf{V}} - \hat{\dot{\textbf{V}}}_{k}+ \frac{1}{\mu_k}\nabla\mathcal{P}(\hat{\dot{\textbf{V}}}_{k})\parallel_F^2.
\end{cases}
\end{equation}
Using the theories of quaternion matrix derivatives, the gradient of $\mathcal{P}(\hat{\dot{\textbf{V}}})$ is 
\begin{equation}
\frac{\partial\mathcal{P}(\dot{\textbf{V}})}{\partial \dot{\textbf{V}}^*}=
\mathbf{W}\odot(\frac{1}{2}\dot{\textbf{V}}\dot{\textbf{U}}_{k+1}^H\dot{\textbf{U}}_{k+1} - \frac{1}{2}\dot{\mathbf{M}}^H\dot{\textbf{U}}_{k+1}).
\end{equation}
Taking the above gradient into (\ref{m13}), we can get
\begin{equation}\label{m14}
\begin{cases}
\hat{\dot{\textbf{V}}}_{k}=\dot{\textbf{V}}_{k}+\omega_k(\dot{\textbf{V}}_{k}-\dot{\textbf{V}}_{k-1}).\\
\dot{\textbf{V}}_{k+1}=\mathop{\arg\min}\limits_{\dot{\textbf{V}}} \frac{\lambda}{2}\parallel\dot{\textbf{V}}\parallel_L^1+ \frac{\mu_k}{2}\parallel \dot{\textbf{V}} - \hat{\dot{\textbf{V}}}_{k}+ \frac{1}{2\mu_k}\mathbf{W}\odot(\hat{\dot{\textbf{V}}}_k\dot{\textbf{U}}_{k+1}^H - \dot{\mathbf{M}}^H)\hat{\dot{\textbf{U}}}_{k+1})\parallel_F^2.
\end{cases}
\end{equation}
Then the QLSVT can be used to solve the above model. The whole steps of the QLNF algorithm are given in Algorithm \ref{a1}.

\begin{algorithm}[htbp]
	\caption{The QLNF  algorithm for quaternion matrix completion}
	\label{a1}
	\begin{algorithmic}[1]
		\REQUIRE   the incomplete matrix data $\dot{\mathbf{M}}\in\mathbb{H}^{M\times N}$, the real matrix $\mathbf{W}$ is used to determine the position of observed and missing elements, $\lambda$, $\mu_{min}$.
		\STATE \textbf{Initial} $\dot{\mathbf{X}}_1=\dot{\mathbf{M}}$,  $t_1$,  $\omega_1$.
		\STATE \textbf{Repeat}
		\STATE Update $\mu_{k+1}^{\dot{\textbf{U}}}=\max \{\parallel\dot{\textbf{V}}_k\parallel_F^2,\mu_{min}\}$.
		\STATE Update $\dot{\textbf{U}}_{k+1}$ by (\ref{m10}).
		\STATE Update $\mu_{k+1}^{\dot{\textbf{V}}}=\max \{\parallel\dot{\textbf{U}}_{k+1}\parallel_F^2,\mu_{min}\}$.
		\STATE Update $\dot{\textbf{V}}_{k+1}$ by (\ref{m14}).
		\STATE $k \longleftarrow k+1.$
		\STATE \textbf{Until convergence}
		\STATE Update $\dot{\textbf{X}}_{opt}=\dot{\textbf{U}}_{k+1}\dot{\textbf{V}}_{k+1}^H.$ 
		\ENSURE  the recovered quaternion matrix.
	\end{algorithmic}
\end{algorithm}
 \subsection{The proposed Truncated Quaternion-based Logarithmic Norm Approximation (TQLNA)}
 As we mentioned in the problem of matrix low rank optimization above, the truncated nuclear norm can achieve a better approximation of the rank function than the nuclear norm by using a two-step iterative scheme \cite{DBLP:journals/pami/HuZYLH13}.  Motivated by this approach, we incorporate the truncated skill and QLN to induce the truncated quaternion-based logarithmic norm as follows
  \begin{definition}(\textbf{TQLN}):
 	\label{def5}
 	Given $\dot{\textbf{X}}\in\mathbb{H}^{M \times N}$, the truncated logarithmic norm of the quaternion matrix with $0\leq p \leq1$ and $\epsilon>0$ is defined as the sum of logarithmic function of  $min(M,N)-r$ minimum singular values: 
 	\begin{equation}
 	\parallel\dot{\textbf{X}}\parallel_{L,r}^p=\sum_{i=r+1}^{min(M,N)}\log(\sigma_{i}^P(\dot{\mathbf{X}})+\epsilon),
 	\end{equation}
 	where $\sigma_{i}$  can be obtained by the QSVD of  $ \dot{\textbf{X}}$.
 \end{definition}
 
 Since the first few largest singular values make no difference to  rank, we will discard them in the TQLN and be committed to optimizing the smallest ${min(M,N)}-r$ singular values  to get a more accurate low rank estimation. According to  the TQLN, the completion based on the low rank minimization model (\ref{m6}) can be formulated as follows
 \begin{equation}\label{m15}
 \min\limits_{\dot{\mathbf{X}}} \parallel\dot{\mathbf{X}}\parallel_{L,r}^p,  \qquad  \text{ s.t.}   \quad   P_\Omega(\dot{\mathbf{X}}-\dot{\mathbf{M}})=0.
 \end{equation}

We have the following theorem based on the Von Neumann's trace inequality \cite{mirsky1975trace} to solve the TQLN. 
\begin{theorem}
	\label{theorem5}
	Given matrix $\dot{\textbf{X}}\in\mathbb{H}^{M \times N}$, and any matrices $\dot{\textbf{A}}\in\mathbb{H}^{r \times M}$  and $\dot{\textbf{B}}\in\mathbb{H}^{r \times N}$ that are satisfied with $\dot{\textbf{A}}\dot{\textbf{A}}^{H}=\textbf{I}_{r\times r}$, $\dot{\textbf{B}}\dot{\textbf{B}}^{H}=\textbf{I}_{r\times r}$. r is any  integer $(r\leq min(M,N))$, we have
	\begin{equation}
	\mid tr(\dot{\textbf{A}}\dot{\mathbf{X}}\dot{\mathbf{B}}^{H})\mid\leq\sum_{i=1}^r\sigma_{i}(\dot{\mathbf{X}}).
	\end{equation}
	Furthermore, 
	$\max |tr(\dot{\mathbf{A}}\dot{\mathbf{X}}\dot{\mathbf{B}}^H)|=\sum_{i=1}^r\sigma_{i}(\dot{\mathbf{X}}).$
\end{theorem}
The proof of Theorem \ref{theorem5} can be found in the Appendix.

Then we can rewrite the model as
 \begin{equation}\label{m16}
\begin{aligned}
&\min\limits_{\dot{\mathbf{X}}} \lambda\parallel\dot{\mathbf{X}}\parallel_{L}^p- \mathop{\max}\limits_{\dot{\mathbf{C}}\dot{\mathbf{C}}^H=\mathbf{I},\dot{\mathbf{D}}\dot{\mathbf{D}}^H=\mathbf{I}} |tr(\dot{\mathbf{C}}\dot{\mathbf{X}}\dot{\mathbf{D}}^H)| ,  \qquad \\& \text{ s.t.}   \quad   P_\Omega(\dot{\mathbf{X}}-\dot{\mathbf{M}})=0.
 \end{aligned}
\end{equation}
This procedure is summarized in Algorithm \ref{a2}.
 \begin{algorithm}[htbp]
 	\caption{The two-step TQLN algorithm}
 	\label{a2}
 	\begin{algorithmic}[1]
 		\REQUIRE   the incomplete matrix data $\dot{\mathbf{M}}\in\mathbb{H}^{M\times N}$, the position set of observed elements $\Omega$, and the tolerance $\varepsilon_0$.
 		\STATE \textbf{Initial} $\dot{\mathbf{X}}_1=\dot{\mathbf{M}}$.
 		\STATE \textbf{Repeat}
 		\STATE \quad \textbf{Step 1.} Given $\dot{\mathbf{X}}_{k}$\\
 		\qquad  \qquad \qquad  $[\dot{\mathbf{U}}_{k},\mathbf{\Sigma}_{k},\dot{\mathbf{V}}_{k}]=QSVD(\dot{\mathbf{X}}_{k})$
 		\STATE  \quad where $\dot{\mathbf{U}}_{k}=(\dot{\mathbf{u}}_{1},\cdots, \dot{\mathbf{u}}_{m})\in\mathbb{H}^{M\times M}$,\\
 		\qquad \quad\quad $\dot{\mathbf{V}}_{k}=(\dot{\mathbf{v}}_{1},\cdots, \dot{\mathbf{v}}_{n})\in\mathbb{H}^{N\times N}$.
 		\STATE   \quad Computing $\dot{\mathbf{C}}_k=(\dot{\mathbf{u}}_{1},\cdots, \dot{\mathbf{u}}_{r})^T\in\mathbb{H}^{r\times M}$ and \\ \qquad\qquad\quad\quad$\dot{\mathbf{D}}_k=(\dot{\mathbf{v}}_{1},\cdots, \dot{\mathbf{v}}_{r})^T\in\mathbb{H}^{r\times N}$.
 		\STATE \quad \textbf{Step 2.} Solving  \\
 		\qquad \quad\quad $\dot{\mathbf{X}}_{k+1}=\arg\min\lambda\parallel\dot{\mathbf{X}}\parallel_L^p-tr(\dot{\mathbf{C}}_k\dot{\mathbf{X}}\dot{\mathbf{D}}_k^T)$
 		\STATE \textbf{Until convergence} $\|\dot{\mathbf{X}}_{k+1}-\dot{\mathbf{X}}_{k}\|_F \leq \varepsilon_0$
 		\ENSURE  the recovered quaternion matrix.
 	\end{algorithmic}
 \end{algorithm}

In the first step, by computing QSVD of  $\dot{\mathbf{X}}_{k}$, we can get $\dot{\mathbf{C}}_k$ and $\dot{\mathbf{D}}_k$. In step 2, we  solve the model (\ref{m16}) by using the ADMM framework as in the complex or quaternion  filed \cite{DBLP:journals/tsp/MiaoK20, li2015alternating}. We first reformulate  model (\ref{m16}) as 
\begin{equation}\label{m17}
\begin{aligned}
&\min\limits_{\dot{\mathbf{X}},\dot{\mathbf{H}}}\lambda\parallel\dot{\mathbf{X}}\parallel_L^p- |tr(\dot{\mathbf{C}}_l\dot{\mathbf{H}}\dot{\mathbf{D}}_l^H)|  \qquad   \\& \text{s.t.}  \quad \dot{\mathbf{X}}=\dot{\mathbf{H}} \quad  P_\Omega(\dot{\mathbf{H}}-\dot{\mathbf{M}})=0,
\end{aligned}
\end{equation}
where $\dot{\mathbf{H}}$ is an intermediate variable. Significantly, the multiplication in the quaternion domain is not commutative, so the Lagrange function of problem  (\ref{m17}) can be written as 
\begin{equation}\label{m18}
\begin{aligned}
L(\dot{\mathbf{X}},\dot{\mathbf{H}},\dot{\mathbf{Y}}, \beta)=&\lambda\parallel\dot{\mathbf{X}}\parallel_L^p- |tr(\dot{\mathbf{C}}_l\dot{\mathbf{H}}\dot{\mathbf{D}}_l^H)|\\&
+\frac{\beta}{2}\parallel\dot{\mathbf{X}}-\dot{\mathbf{H}}\parallel_F^2\\&
+\mathfrak{R}(tr(\dot{\mathbf{Y}}^H(\dot{\mathbf{X}}-\dot{\mathbf{H}}))),
\end{aligned}
\end{equation}
where $\beta_k$ is a positive penalty parameter,  and $\dot{\mathbf{Y}}$ is the Lagrange multiplier.

\textbf{Updating $\dot{\mathbf{X}}$}:
Assuming $\dot{\mathbf{X}}_{\tau +1}$ is $\tau +1$-\textit{th} iteration result in step 2 and keeping other variables with the latest value, the optimal solution of  $\dot{\mathbf{X}}_{\tau +1}$ can be obtained from the following problem
\begin{equation}\label{m19}
\begin{aligned}
\dot{\mathbf{X}}_{\tau +1}&=\arg\min_{\dot{\mathbf{X}}}\lambda\parallel\dot{\mathbf{X}}\parallel_L^p
 +\frac{\beta_\tau}{2} \parallel\dot{\mathbf{X}}-\dot{\mathbf{H}}_\tau \parallel_F^2
+\mathfrak{R}(tr(\dot{\mathbf{Y}}^H_\tau (\dot{\mathbf{X}}-\dot{\mathbf{H}}_\tau )))\\&
=\arg\min_{\dot{\mathbf{X}}}\lambda\parallel\dot{\mathbf{X}}\parallel_L^p+
\frac{\beta_\tau}{2}\parallel\dot{\mathbf{X}}-(\dot{\mathbf{H}}_\tau-\frac{1}{\beta_\tau}\dot{\mathbf{Y}}_\tau)\parallel_F^2
\end{aligned}
\end{equation}
Utilizing the QLSVT,  when  let $p=1$, $\dot{\mathbf{A}}=\dot{\mathbf{H}}_\tau-\frac{1}{\beta_\tau}\dot{\mathbf{Y}}_\tau$, we can get 
\begin{equation}\label{m20}
\dot{\mathbf{X}}_{\tau +1}=\dot{\mathbf{U}}_{\dot{\textbf{A}}}\mathcal{L}_{\frac{\lambda}{\beta_\tau}, \epsilon}(\mathbf{\Lambda}_{\dot{\textbf{A}}})\dot{\mathbf{V}}_{\dot{\textbf{A}}}^H
\end{equation}

\textbf{Updating $\dot{\mathbf{H}}$}:
Fixing $\dot{\mathbf{X}}_{\tau+1}$ and $\dot{\mathbf{Y}}_\tau$, the optimal solution of $\dot{\mathbf{H}}_{\tau +1}$ can be obtained from the following problem
\begin{equation}\label{m21}
\begin{aligned}
\dot{\mathbf{H}}_{\tau+1}&=\arg\min_{\dot{\mathbf{H}}}- |tr(\dot{\mathbf{C}}_l\dot{\mathbf{H}}\dot{\mathbf{D}}_l^H)|
+\frac{\beta_\tau}{2}\parallel\dot{\mathbf{X}}_{\tau+1}-\dot{\mathbf{H}}\parallel_F^2
+\mathfrak{R}(tr(\dot{\mathbf{Y}}^H_\tau(\dot{\mathbf{X}}_{\tau+1}-\dot{\mathbf{H}})))
\\&=\arg\min_{\dot{\mathbf{H}}}\frac{\beta_\tau}{2}\|\dot{\mathbf{H}}-(\dot{\mathbf{X}}_{\tau+1}+
\frac{1}{\beta_\tau}(\dot{\mathbf{C}}_l^H\dot{\mathbf{D}}_l+\dot{\mathbf{Y}}_\tau))\|^2_F.
\end{aligned}
\end{equation}
The closed  solution of  $\dot{\mathbf{H}}_{\tau+1}$ can be attained by
\begin{equation}\label{m22}
\dot{\mathbf{H}}_{\tau+1}=\dot{\mathbf{X}}_{\tau+1}+\frac{1}{\beta_\tau}(\dot{\mathbf{C}}_l^H\dot{\mathbf{D}}_l+\dot{\mathbf{Y}}_\tau).
\end{equation}
Then, we let the values of all observed elements be constant in each iteration and reach
\begin{equation}\label{m23}
\dot{\mathbf{H}}_{\tau+1}=P_{\Omega^C}(\dot{\mathbf{H}}_{\tau+1})+P_\Omega(\dot{\mathbf{M}}).
\end{equation}

\textbf{ Updating $\dot{\mathbf{Y}}_{\tau+1}$}: Fixing $\dot{\mathbf{X}}_{\tau+1}$ and $\dot{\mathbf{H}}_{\tau+1}$, the optimal solution of $\dot{\mathbf{Y}}_{\tau +1}$ can be obtained from the following problem:
\begin{equation}\label{model15}
\dot{\mathbf{Y}}_{\tau+1}=\dot{\mathbf{Y}}_{\tau}+\beta_\tau(\dot{\mathbf{X}}_{\tau+1}-\dot{\mathbf{H}}_{\tau+1}).
\end{equation}

\textbf{Updating the penalty parameter $\beta_{\tau+1}$}:
\begin{equation}
\beta_{\tau+1}=\rho\beta_{\tau}.
\end{equation}
The whole procedure in step 2 is summarized in Algorithm \ref{a3}.
 \begin{algorithm}[htbp]
	\caption{The optimization of TQLN using ADMM in step 2}
	\label{a3}
	\begin{algorithmic}[1]
		\REQUIRE    $\dot{\mathbf{M}}$, $\Omega$, $\dot{\mathbf{C}}_l$, $\dot{\mathbf{D}}_l$, and the tolerance $\varepsilon$, $\lambda$, $\rho$, $\beta_{max}$.
		\STATE \textbf{Initial} $\dot{\mathbf{X}}_1=\dot{\mathbf{M}}$, $\dot{\mathbf{H}}_1=\dot{\mathbf{X}}_1$, $\dot{\mathbf{Y}}_1=\dot{\mathbf{X}}_1$, and $\beta_1$.
		\STATE \textbf{Repeat}
		\STATE Update $\dot{\mathbf{X}}_{\tau+1}=\dot{\mathbf{U}}_{\dot{\textbf{A}}}\mathcal{L}_{\frac{\lambda}{\beta_\tau}, \epsilon}(\mathbf{\Lambda}_{\dot{\textbf{A}}})\dot{\mathbf{V}}_{\dot{\textbf{A}}}^H$.
		\STATE  Update $\dot{\mathbf{H}}_{\tau +1}=\dot{\mathbf{X}}_{\tau+1}+\frac{1}{\beta_\tau}(\dot{\mathbf{C}}_l^H\dot{\mathbf{D}}_l+\dot{\mathbf{Y}}_\tau)$
		\STATE   Update $\dot{\mathbf{Y}}_{\tau+1}=\dot{\mathbf{Y}}_{\tau}+\beta_\tau(\dot{\mathbf{X}}_{\tau+1}-\dot{\mathbf{H}}_{\tau+1}).$
		\STATE Update  $\beta_{\tau+1}=\min(\rho\beta_{\tau}, \beta_{max})$ 
		\STATE $\tau\longleftarrow \tau+1$
		\STATE \textbf{Until convergence} $\|\dot{\mathbf{X}}_{\tau+1}-\dot{\mathbf{X}}_{\tau}\|_F \leq \varepsilon$
		\ENSURE  $\dot{\mathbf{X}}_{\tau+1}$, $\dot{\mathbf{H}}_{\tau +1}$, $\dot{\mathbf{Y}}_{\tau+1}$.
	\end{algorithmic}
\end{algorithm}

\section{Experimental results}\label{Experimental results}
In this section, we verify the performance of the proposed TQNF and TQLNA methods. We compare  relevant and  state-of-art algorithms:
\begin{enumerate}
  \item \textbf{WNNM} \cite{DBLP:conf/cvpr/GuZZF14}: the weighted nuclear norm of the real matrix is used to depict low rank.
  \item \textbf{TNNR} \cite{DBLP:journals/pami/HuZYLH13}: the truncated nuclear norm of the real matrix is used to avoid the largest several singular values being punished much.
  \item \textbf{D-N} \cite{DBLP:journals/pami/ShangCLLL18}: the real bilinear factor matrices are proposed to  depict low rank based on nuclear norm.
  \item \textbf{F-N} \cite{DBLP:journals/pami/ShangCLLL18}: the real bilinear factor matrices are proposed to  depict low rank based on  Frobenius norm.
  \item \textbf{LRMF} \cite{DBLP:journals/tip/ChenJLZ21}: the real factor matrices are proposed to depict low rank based on  Logarithmic norm.
  \item \textbf{LRQA-2} \cite{DBLP:journals/tip/ChenXZ20}:  the Laplace function is used to replace QNN to depict the low rank of quaternion matrix.
  \item \textbf{Q-DNN} \cite{DBLP:journals/tsp/MiaoK20}:  bilinear factor quaternion matrices factorization is designed based on QNN.
   \item \textbf{Q-FNN} \cite{DBLP:journals/tsp/MiaoK20}: bilinear factor quaternion matrices factorization is designed based on quaternion Frobenius norm.
\end{enumerate}

All the experiments are operated in Matlab R2019a, on a PC with a 3.00GHz CPU and RAM of 8GB.

\textbf{Parameters setting:} For QLNF, we let $t_1=1$, $\omega_{0}=0$, $\lambda=1.25e-5$, and $\mu_{\min}=0.005$. The stopping criterion is defined as $means(\|\dot{\mathbf{U}}_{k+1}-\dot{\mathbf{U}}_{k}\|_F /\|\dot{\mathbf{U}}_k\|_F,\|\dot{\mathbf{V}}_{k+1}-\dot{\mathbf{V}}_{k}\|_F /\|\dot{\mathbf{V}}_k\|_F) \leq \epsilon_0$, where $\epsilon_0=0.001$. The appropriate $d$ is chosen from $\{6, 8, 10, 20\}$. For TQLNA, due to the  lack of prior knowledge of  the number of  truncated singular values, letting  $r = 1$ to report the experimental results, the better results can be achieved with appropriate adjustments of $r$. We let $\rho=1.5$, $\beta_{max}=10^7$, and $\beta_0=0.003$. The stopping criterion is $\|\dot{\mathbf{X}}_{k+1}-\dot{\mathbf{X}}_{k}\|_F /\|\dot{\mathbf{M}}\|_F\leq \epsilon_0$, where $\epsilon_0=0.001$. 
The parameters of these compared methods are set as the experimental settings reported in their papers individually. Eight frequently used color images are selected as  test samples, which  are shown in Figure  \ref{Ytu}. 
\begin{figure}
	\centering
	\includegraphics[width=100mm]{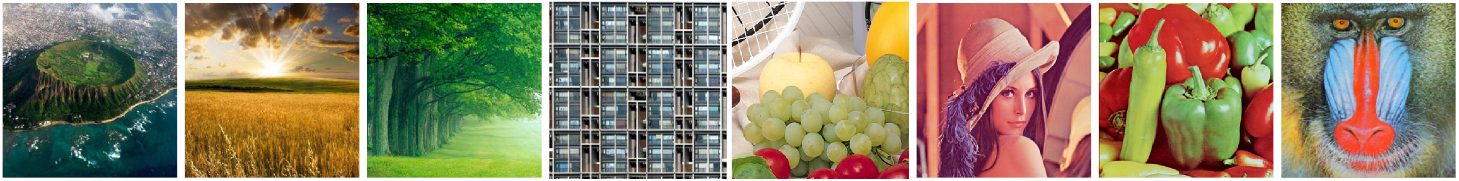}
	\caption{The $8$ color images (from left to right, Image(1) $300\times 300\times 3$,
		Image(2)  $300\times 300\times 3$, Image(3)  $300\times 300\times 3$, Image(4)  $300\times 300\times 3$, Image(5)  $300\times 300\times 3$, Image(6)  $300\times 300\times 3$, Image(7)  $300\times 300\times 3$, and Image(8)  $300\times 300\times 3$).}
	\label{Ytu}
\end{figure}

\textbf{Performance index setting:} The peak signal to noise rate (PSNR)  and the structural similarity index (SSIM)  are adopted to compare the quality of  reconstruction results.   The best results are \textbf{bolded} in experiments. For quaternion-based approaches, we load the test image as a whole quaternion matrix.  Matrix-based approaches are performed on each channel of the test image respectively. The smaller value of  Sampling Rate (SR) represents the higher degree of missing. 

Table \ref{t1} and Table \ref{t2} display the results of recovering these 8 images. The PSNR, SSIM, and average CPU time (seconds) of all testing images reconstructed by 10 utilized recovering methods for different SRs (10\%, 15\%, 25\%, 35\%, 45\%, and 50\%, respectively).  Figure \ref{0.85} displays the corresponding visual  comparisons of  SR $= 15\%$ between the  two designed methods and other methods of comparison on the above eight tested color images.

Figure \ref{baboon} compares  the visual results among all competing inpainting methods on  recovering image(8) with SR $= 20\%$. The corresponding PSNR and SSIM are given in Table \ref{babooni}. By comparing PSNR values, we also test the robustness of TQLNA with different truncated number  $r\in[1, 20]$ when SR $= 40\%$ in Figure \ref{r}. 

In addition to the above eight images, we randomly select 60 color images from Berkeley Segmentation Dataset (BSD) \footnote{Available: https://www2.eecs.berkeley.edu/Research/Projects/CS/vision/bsds/} to fully explain the effectiveness of the proposed algorithms. In Figure \ref{0.8}, we map out the PSNR and 
SSIM values on the 60 color images with SR $= 20\%$. The average values are reported in Table \ref{0.8i}. Similarly, in Figure \ref{0.6}, we map out the PSNR and
SSIM values on the 60 color images with SR $= 40\%$. The average values are reported in Table \ref{0.6i}.  All the straight lines  represent the average results obtained by ten methods. 

From the above experimental results, we can draw the following summaries.
\begin{itemize}
	\item When testing the eight images, we set SR$ = 10\%$ as an extreme case. In these experiments, we find that our methods are able to get higher PSNR and SSIM values in most cases (Six-eighths) without changing the truncated number $r$ of TQLNA method. For QLNF method, it also attains better parameter results when the missing rate is high. On the contrary, when  SR$ = 50\%$, the TQLNA method is the  most effective way to restore images from the respect of  PSNR and SSIM values.  In Figure \ref{baboon}, it is observed that TQLNA  recovery more details of  the zoomed-in section, and Table \ref{babooni} further confirms this conclusion. This is mainly as a result of the superiority from the quaternion representation of color pixel values.
	\item  When  using the BSD datasets, we set SR$ = 20\%$ and SR$ = 40\%$ and draw the PSNR, SSIM values and average values of each picture. The TQLNA method is able to achieve better results in almost all cases, but the effectiveness of QLNF method would be slightly reduced when the degree of  absence is lower. Besides, the TNNR method can also reach superior results. Profiting from this truncated skill, our quaternion-based method  is optimized.
	\item For the time consumption, the matrix-based factorization methods like D-N, F-N are much more time-saving than others. This is due to only two  smaller factor matrices need to be handled. Although our method is a bit time-consuming, it still reduces the time cost in the quaternion domain, especially for TQLNA method.
\end{itemize} 

\renewcommand{\arraystretch}{1.5}
\begin{table*}
	\centering
	\fontsize{6.5}{8}\selectfont
	\begin{threeparttable}
		\caption{THE PSNR and  SSIM  obtained by different completion algorithms for 8 color images (image(1)-image(4))}
		\scalebox{0.8}{
			\begin{tabular}{ccccccccccccccc}
				\toprule
				\multirow{2}{*}{Image}
				& SR& \multicolumn{2}{c}{10$\%$}&\multicolumn{2}{c}{15$\%$}&\multicolumn{2}{c}{25$\%$}&\multicolumn{2}{c}{35$\%$}&\multicolumn{2}{c}{45$\%$}&\multicolumn{2}{c}{50$\%$}&\multirow{2}{*}{Aver. Time}\cr
				\cmidrule(lr){3-4}\cmidrule(lr){5-6}\cmidrule(lr){7-8}
				\cmidrule(lr){9-10}\cmidrule(lr){11-12}\cmidrule(lr){13-14}
				&Method&PSNR&SSIM&PSNR&SSIM&PSNR&SSIM&PSNR&SSIM&PSNR&SSIM&PSNR&SSIM\cr
				\midrule
				\multirow{10}*{Image(1)}
				&TNNR%
				
				&18.316&0.408&19.469&0.512&21.265&0.628
				&22.674&0.713&24.135&	0.787&24.949&0.819&98.122\\
				&WNNR
				&14.676&0.283&16.538&0.332&18.424&0.451	      &20.110&0.563&21.584&0.656&22.402&0.701&37.956\\
				&D-N
				&17.929	&0.3694&18.831&0.429&20.139&0.512	 &20.697&0.545&20.924&0.559&21.050&0.567&1.889\\
				&F-N
				&18.289&0.385&19.063&0.438&20.225&0.514
				&20.716&0.545&20.954&0.559&21.054&0.568&\textbf{1.244}\\
				&LRMF
				&17.830&0.376&18.444&0.427&18.790&0.460	 &20.643&0.572&22.200&0.652&22.697&0.678&1.451\\
				&LRQA-2
				&18.788&0.448&19.645&0.514&21.163&0.615 &22.598&0.703&24.052&0.778&24.876&0.811&81.105\\
				&Q-DNN
				&18.529&0.436&19.334&0.486&20.773&0.583
				&22.091&0.664&23.261&0.724&23.900&0.753&28.717\\
				&Q-FNN
				&18.779&0.448&19.546&0.496&20.808&0.576 &22.130&0.654&23.385&0.727&24.515&0.786&30.628\\
				&QLNF
				&18.668&0.427&19.443&0.482&20.699&0.560
				&21.525&0.606&22.461&0.662&22.955&0.689&39.814\\
				&TQLNA
				&\textbf{18.817}&\textbf{0.453}&\textbf{19.831}&\textbf{0.529}&\textbf{21.357}&\textbf{0.630}	 &\textbf{22.766}&\textbf{0.713}&\textbf{24.260}&\textbf{0.786}&\textbf{25.053}&\textbf{0.817}&28.507\\
																
				\cline{1-15}\multirow{10}*{Image(2)}			
				&TNNR%
				
				&20.630&0.813&21.951&0.856&23.629&0.898
				&25.226&0.929&26.745&0.950&27.489&0.959&24.005\\
				&WNNR
				&14.087&0.572&19.136&0.749&21.282&0.832	      &23.042&0.886&24.543&0.920&25.269&0.934&35.459\\
				&D-N
				&19.985&0.793&21.140&0.830&22.260&0.864	 &22.755&0.877&22.989&0.884&23.068&0.886&1.673\\
				&F-N
				&20.251&0.807&21.220&0.834&22.381&0.866	 &22.772&0.877&22.991&0.883&23.071&0.886&\textbf{1.106}\\
				&LRMF
				&20.111&0.794&21.996&0.854&22.432&0.872	 &23.602&0.900&24.611&0.920&25.101&0.927&1.324\\
				&LRQA-2
				&20.863&0.821&21.852&0.851&23.564&0.894	 &25.072&0.924&26.505&0.945&27.253&0.954&76.154\\
				&Q-DNN
				&20.606&0.816&21.996&0.856&23.385&0.892
				&24.935&0.923&25.980&0.938&26.584&0.945&27.809	\\
				&Q-FNN 
				&20.775&0.819&21.996&0.856&23.466&0.894 &24.745&0.919&25.715&0.935&26.468&0.945&29.920\\
				&QLNF
				&20.732&0.819&21.666&0.846&22.805&0.879
				&23.363&0.892&24.114&0.909&24.302&0.913&40.087\\
				&TQLNA
				&\textbf{20.955}&\textbf{0.822}&\textbf{22.079}&\textbf{0.857}&
				\textbf{23.734}&\textbf{0.899	}	 &\textbf{25.390}&\textbf{0.930}&\textbf{26.881}&\textbf{0.951}&
				\textbf{27.680}&\textbf{0.961}&	23.756\\
												
				\cline{1-15}\multirow{10}*{Image(3)}
				&TNNR%
				
				&21.291&0.807&22.957&0.856&24.380&0.893
				&25.794&0.920&27.183&0.941&27.867&0.949&108.649\\
				&WNNR
				&16.636&0.619&19.984&0.737&21.707&0.808      &23.351&0.862&24.794&0.897&25.548&0.912&37.360\\
				&D-N
				&21.358&0.808&22.337&0.838&23.312&0.867 &23.775&0.879&24.021&0.885&24.090&0.887&1.799\\
				&F-N
				&21.676&0.820&22.590&0.846&23.437&0.870
				&23.819&0.880&24.011&0.885&24.088&0.887&\textbf{1.144}\\
				&LRMF
				&21.321&0.804&22.291&0.830&23.050&0.854	 &24.288&0.885&25.405&0.908&25.803&0.915&1.369\\
				&LRQA-2
				&22.155&0.836&23.053&0.859&24.235&0.890	 &25.643&0.918&27.028&0.938&27.698&0.947&80.198\\
				&Q-DNN
				&21.798&0.827&22.818&0.850&24.029&0.882
				&25.355&0.910&26.376&0.927&26.900&0.934&28.708\\
				&Q-FNN
				&22.154&0.837&23.108&0.860&24.194&0.886 &25.236&0.907&26.363&0.926&26.919&0.934&30.493\\
				&QLNF
				&21.848&0.837&22.683&0.856&23.397&0.872
				&24.116&0.887&24.452&0.893&24.558&0.896&32.660\\
				&TQLNA
				&\textbf{22.318}&\textbf{0.845}&\textbf{23.286}&\textbf{0.868}&
				\textbf{24.578}&\textbf{0.898}	 &\textbf{25.943}&\textbf{0.923}&\textbf{27.385}&\textbf{0.943}&
				\textbf{28.076}&\textbf{0.951}&	18.037\\									
				\cline{1-15}\multirow{10}*{Image(4)}
				&TNNR%
				
				&17.356&0.722&19.558&0.814&22.664&0.894
				&25.037&0.934&27.159&0.958&28.293&0.968&43.968\\
				&WNNR
				&11.990&0.392&17.031&0.683&21.572&0.856	      &24.269&0.917&26.275&0.945&27.231&0.956&35.826\\
				&D-N
				&17.495&0.712&19.399&0.799&20.661&0.844	 &21.055&0.858&21.323&0.866&21.404&0.868&1.928\\
				&F-N
				&17.685&0.717&19.428&0.801&20.588&0.840	 &21.062&0.858&21.312&0.865&21.378&0.865&\textbf{1.222}	\\
				&LRMF
				&17.023&0.686&19.151&0.782&22.464&0.881	 &24.763&0.925&26.008&0.942&26.358&0.947&1.269\\
				&LRQA-2
				&17.962&0.743&\textbf{20.135}&\textbf{0.825}&23.281&0.903	 &25.521&0.939&27.578&0.960&28.577&0.968&77.018\\
				&Q-DNN
				&17.894&0.742&19.913&0.813&23.290&0.901
				&25.697&0.940&27.690&0.960&28.534&0.967&28.371	\\
				&Q-FNN
				&17.767&0.733&19.922&0.816&\textbf{23.508}&\textbf{0.907}
				&25.616&0.939&26.849&0.953&27.627&0.960&	30.187\\
				&QLNF
				&17.881&0.739&19.764&0.816&22.461&0.887
				&23.831&0.913&24.511&0.923&24.995&0.931&38.249	\\
				&TQLNA
				&\textbf{18.023}&\textbf{0.744}&20.100&\textbf{0.825}&
				23.380&0.904	 &\textbf{25.843}&\textbf{0.942}&\textbf{28.165}&\textbf{0.965}&
				\textbf{29.210}&\textbf{0.973}&36.940\\
				\bottomrule
				\label{t1}
		\end{tabular}}
	\end{threeparttable}
\end{table*}

\renewcommand{\arraystretch}{1.5}
\begin{table*}
	\centering
	\fontsize{6.5}{8}\selectfont
	\begin{threeparttable}
		\caption{THE PSNR and  SSIM  obtained by different completion algorithms for 8 color images (image(5)-image(8))}
		\scalebox{0.8}{
			\begin{tabular}{ccccccccccccccc}
				\toprule
				\multirow{2}{*}{Image}
				& SR& \multicolumn{2}{c}{10$\%$}&\multicolumn{2}{c}{15$\%$}&\multicolumn{2}{c}{25$\%$}&\multicolumn{2}{c}{35$\%$}&\multicolumn{2}{c}{45$\%$}&\multicolumn{2}{c}{50$\%$}&\multirow{2}{*}{Aver. Time}\cr
				\cmidrule(lr){3-4}\cmidrule(lr){5-6}\cmidrule(lr){7-8}
				\cmidrule(lr){9-10}\cmidrule(lr){11-12}\cmidrule(lr){13-14}
				&Method&PSNR&SSIM&PSNR&SSIM&PSNR&SSIM&PSNR&SSIM&PSNR&SSIM&PSNR&SSIM\cr
				\midrule
		\multirow{10}*{Image(5)}
			&TNNR%
			
			&16.944&0.625&19.200&0.700&21.326&0.781
			&23.148&0.837&24.664&0.874&25.308&0.891&67.839	\\
			&WNNR
			&11.828&0.370&16.773&0.556&19.286&0.670	      &21.185&0.747&22.735&0.797&23.327&0.814&34.571\\
			&D-N
			&17.532&0.614&18.622&0.674&19.709&0.730	 &20.294&0.759&20.559&0.771&20.619&0.775&1.778\\
			&F-N
			&17.751&0.630&18.760&0.687&19.782&0.737	 &20.312&0.761&20.599&0.774&20.604&0.776&\textbf{1.181}\\
			&LRMF
			&17.166&0.600&18.234&0.641&19.731&0.701	 &21.722&0.773&23.002&0.811&23.402&0.825&1.358\\
			&LRQA-2
			&18.264&0.663&19.598&0.715&21.459&0.781 &23.209&0.833&24.609&0.867&25.277&0.882&69.690\\
			&Q-DNN
			&17.940&0.635&19.240&0.688&21.279&0.765
			&23.005&0.818&24.220&0.847&24.730&0.858&28.348	\\
			&Q-FNN
			&\textbf{18.484}&\textbf{0.672}&19.649&0.716&21.495&0.777 &23.031&0.821&24.454&0.855&25.016&0.869&	30.952\\
			&QLNF
			&18.115&0.649&19.212&0.697&21.063&0.768
			&22.258&0.805&23.090&0.827&23.247&0.835&	43.310\\
			&TQLNA
			&18.411&0.668&\textbf{19.767}&\textbf{0.723}&\textbf{21.788}&\textbf{0.797} &\textbf{23.670}&\textbf{0.850}&\textbf{25.174}&\textbf{0.883}&\textbf{25.822}&\textbf{0.898}&33.542\\
						
			\cline{1-15}\multirow{10}*{Image(6)}
			&TNNR%
			
			&18.229&	0.789&20.126&0.849&22.258&0.897
			&24.142&0.928&25.989&0.950&26.842&0.959&21.895\\
			&WNNR
			&12.955&0.542&17.738&0.735&20.348&0.829	      &22.417&0.886&24.082&0.918&24.955&0.932&38.762\\
			&D-N
			&18.025&0.778&19.787&0.840&20.950&0.874 &21.362&0.887&21.599&0.893&21.691&0.897&1.783\\
			&F-N
			&18.203&0.794&19.953&0.848&20.979&0.878	 &21.369&0.889&21.613&0.895&21.681&	0.896&\textbf{1.117}\\
			&LRMF
			&17.583&0.757&19.370&0.813&20.830&0.851
			&22.952&0.901&24.327&0.926&24.796&0.933&1.302\\
			&LRQA-2
			&\textbf{18.579}&\textbf{0.804}&20.394&0.852&22.502&0.898 &24.249&0.927&25.990&0.948&26.861&0.957&74.180\\
			&Q-DNN
			&18.508&0.795&20.198&0.842&22.380&0.892
			&24.180&0.924&25.579&0.942&26.195&0.949&28.342		\\
			&Q-FNN
			&18.528&0.801&20.331&0.849&22.502&0.895 &24.011&0.921&25.496&0.941&26.206&0.949&30.626\\
			&QLNF
			&18.444&0.801&20.227&0.850&22.053&0.891
			&23.063&0.910&23.834&0.922&24.369&0.930&39.760\\
			&TQLNA
			&18.439&0.796&\textbf{20.507}&\textbf{0.857}&\textbf{22.767}&\textbf{0.904	}
			&\textbf{24.625}&\textbf{0.933}&\textbf{26.488}&\textbf{0.954}&\textbf{27.309}&\textbf{0.961}&29.378\\						
			
			\cline{1-15}\multirow{10}*{Image(7)}
			&TNNR%
			&16.956&0.750&18.966&0.824&21.550&0.889
			&23.858&0.929&25.779&0.952&26.804&0.962&45.206\\
			&WNNR
			&12.796&0.507&16.699&0.713&19.888&0.834      &22.302&0.896&24.227&0.931&25.011&0.941&	36.000\\
			&D-N
			&17.164&0.752&18.564&0.810&19.422&0.840	 &19.874&0.855&20.060&0.861&20.152&0.864&1.980\\
			&F-N
			&17.217&0.756&18.602&0.813&19.456&0.842	 &19.830&0.854&20.059&0.861&20.137&0.863&\textbf{1.213}\\
			&LRMF
			&16.363&0.716&18.153&0.786&20.773&0.864	 &23.113&0.914&24.342&0.933&24.762&0.939&1.269	\\
			&LRQA-2
			&17.419&0.769&19.568&0.838&22.140&0.898	 &24.402&0.935&26.301&0.957&27.298&0.965	&77.539\\
			&Q-DNN
			&17.187&0.752&19.301&0.824&21.992&0.892
			&24.289&0.932&25.852&0.952&26.517&0.958&28.411		\\
			&Q-FNN
			&17.505&0.774&19.788&0.844&22.474&0.904 &24.285&0.933&25.898&0.953&26.816&0.961&31.637\\
			&QLNF
			&17.237&0.763&19.483&0.835&21.682&0.889
			&23.331&0.920&24.168&0.933&24.063&0.931&36.881	\\
			&TQLNA
			&\textbf{17.549}&\textbf{0.774}&\textbf{19.826}&\textbf{0.848}&\textbf{22.528}&\textbf{0.908}	 
			&\textbf{25.000}&\textbf{0.944}&\textbf{26.877}&\textbf{0.963}&\textbf{27.926}&\textbf{0.971}&31.453\\
			\cline{1-15}\multirow{10}*{Image(8)}
			&TNNR%
			
			&18.923&0.548&20.234&0.612&22.100&0.713
			&23.682&0.786&25.318&0.848&26.192&0.873&20.206\\
			&WNNR
			&13.683&0.305&17.220&0.460&19.279&0.577	      &21.084&0.680&22.934&0.767&23.726&0.798	&35.519\\
			&D-N
			&18.405&0.515&19.859&0.574&21.084&0.630 &21.437&0.646&21.663&0.657&21.716&0.658&1.696\\
			&F-N
			&18.767&0.527&19.999&0.581&21.118&0.631	 &21.454&0.645&21.664&0.656&21.712&0.658&\textbf{1.088}	\\
			&LRMF
			&18.278&0.501&19.203&0.552&20.115&0.614	 &21.893&0.698&23.342&0.752&23.759&0.766&1.331\\
			&LRQA-2
			&19.211&0.559&20.341&0.611&22.036&0.706	 &23.615&0.780&25.271&0.842&26.086&0.865&73.468\\
			&Q-DNN
			&18.895&0.535&19.939&0.589&21.678&0.686
			&23.037&0.750&24.467&0.807&24.992&0.824&28.448	\\
			&Q-FNN
			&19.176&0.550&20.290&0.603&21.944&0.686 &23.247&0.748&24.790&0.812&26.060&0.859&31.406\\
			&QLNF
			&19.190&0.548&20.284&0.592&21.727&0.663
			&22.571&0.700&22.956&0.715&23.314&0.732&37.021\\
			&TQLNA
			&\textbf{19.396}&\textbf{0.575}&\textbf{20.690}&\textbf{0.631}&\textbf{22.453}&\textbf{0.724}	 &\textbf{23.964}&\textbf{0.793}&\textbf{25.726}&\textbf{0.856}&\textbf{26.565}&\textbf{0.879}&26.266\\
				\bottomrule
				\label{t2}
		\end{tabular}}
	\end{threeparttable}
\end{table*}
\begin{figure}
	\centering
	\includegraphics[width=14cm,keepaspectratio]{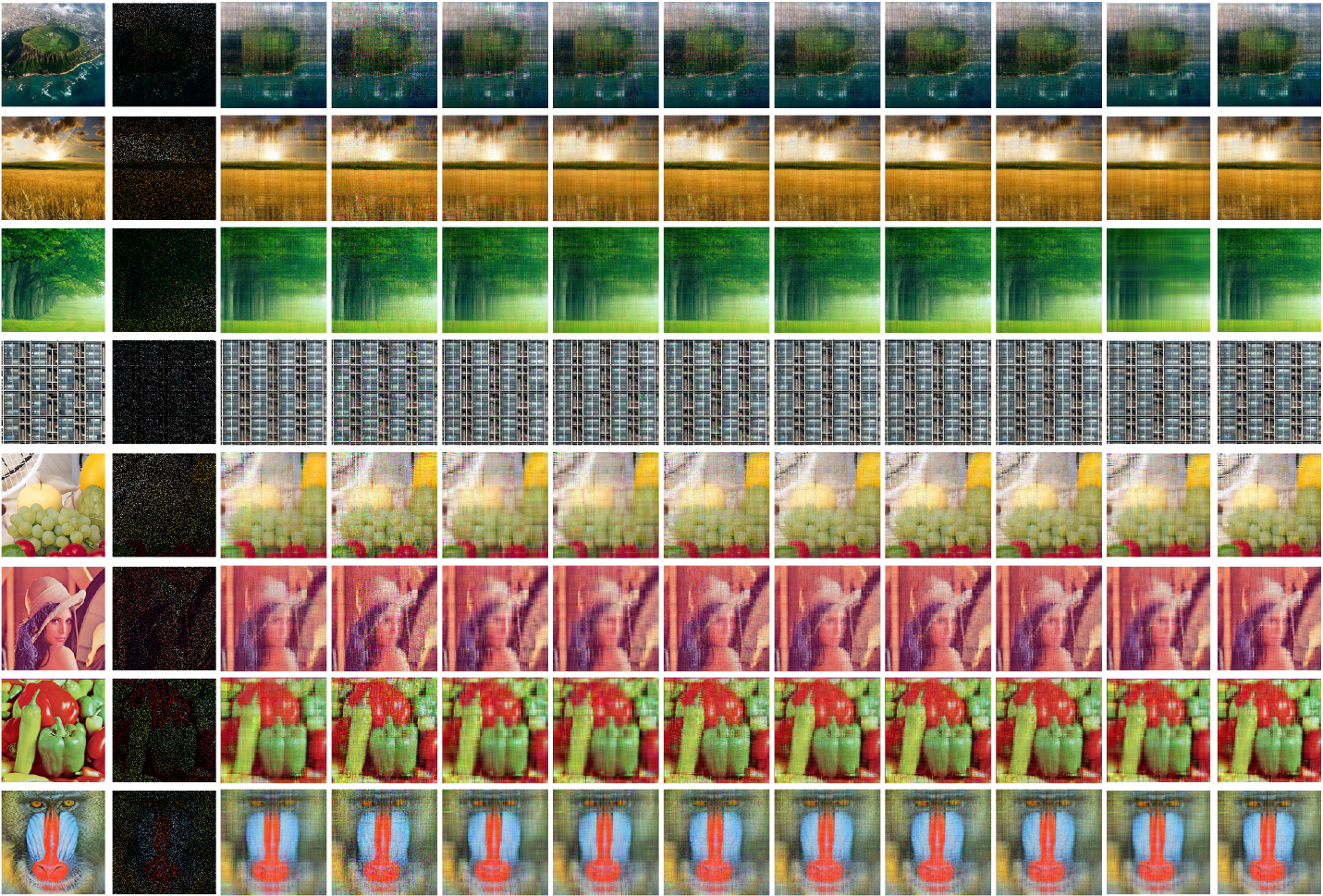}
	\caption{The first column is the original image and the second column is the observed image ( SR=15\%). The 3-\textit{th} to 12-\textit{th} are the completion results of TNNR, WNNR, D-N, F-N, LRMF, LRQA-2, Q-DNN, Q-FNN, QLNF , and TQLNA, respectively. The corresponding parameters are summarized in Table \ref{t1} and Table \ref{t2}.}
	\label{0.85}
\end{figure}
\begin{figure*}
	\centering
	\subfigure[Baboon]{
		\includegraphics[width=3cm]{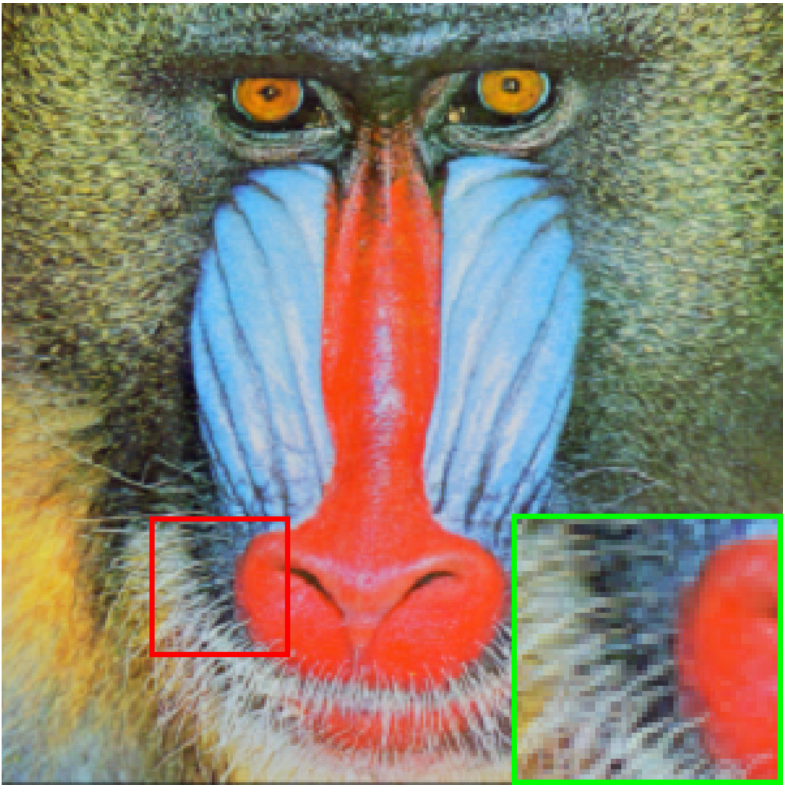}
	}
	\hspace{-0.5cm}
	\subfigure[SR = 20\%]{
		\includegraphics[width=3cm]{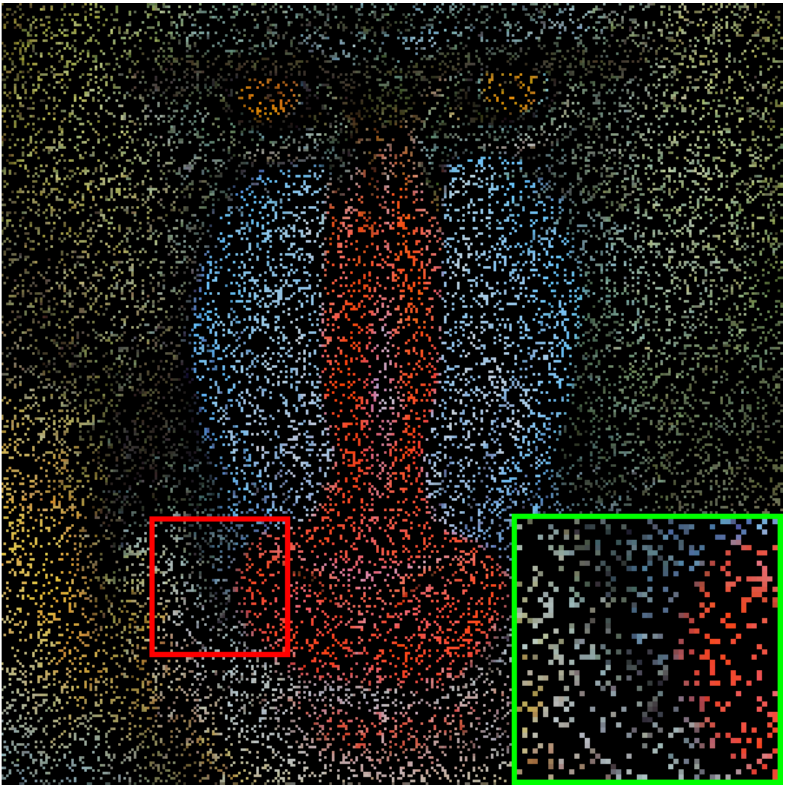}
	}
	\hspace{-0.5cm}
	\subfigure[TNNR]{
		\includegraphics[width=3cm]{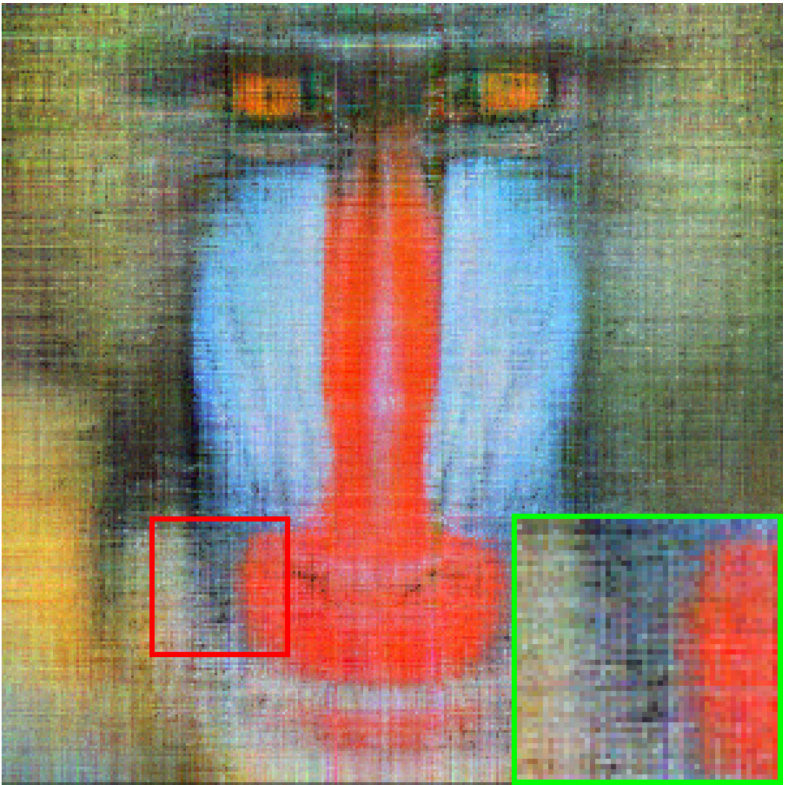}
	}
	\hspace{-0.5cm}
	\subfigure[WNNR]{
		\includegraphics[width=3cm]{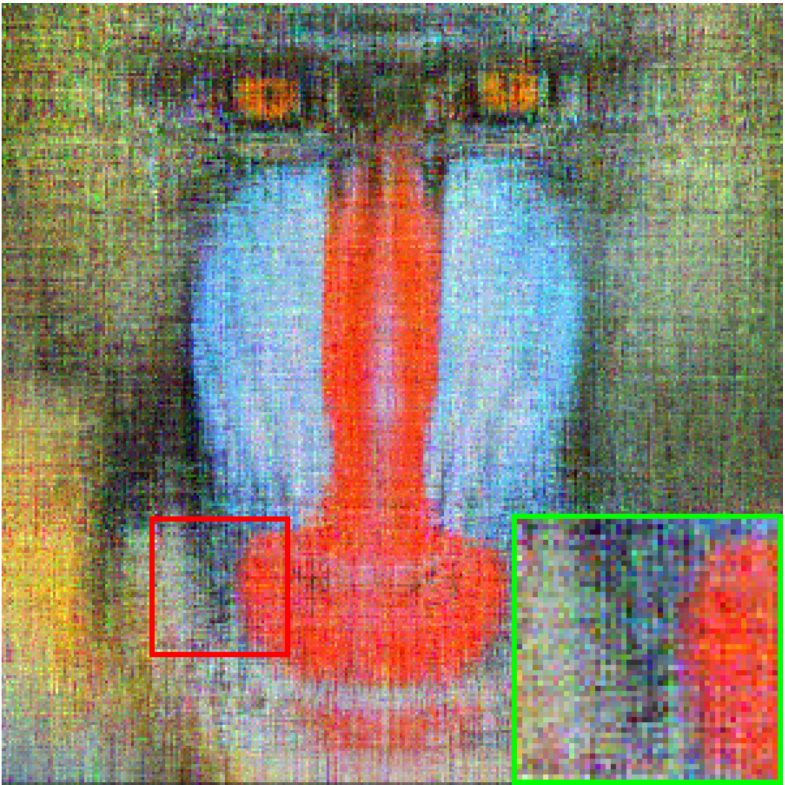}
	}
\\
	\centering
\subfigure[D-N]{
	\includegraphics[width=3cm]{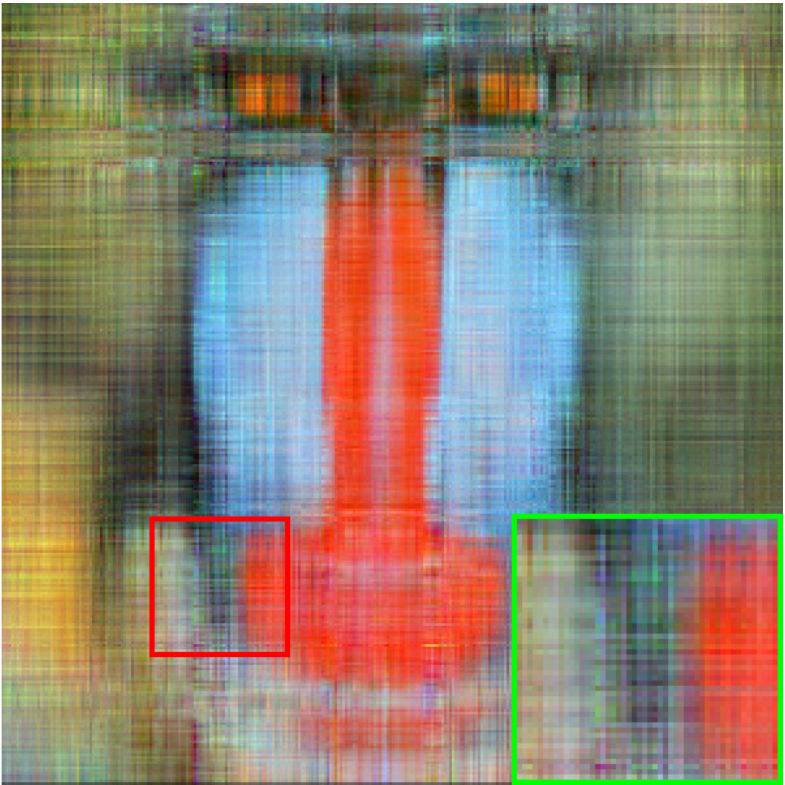}
}
\hspace{-0.5cm}
\subfigure[F-N]{
	\includegraphics[width=3cm]{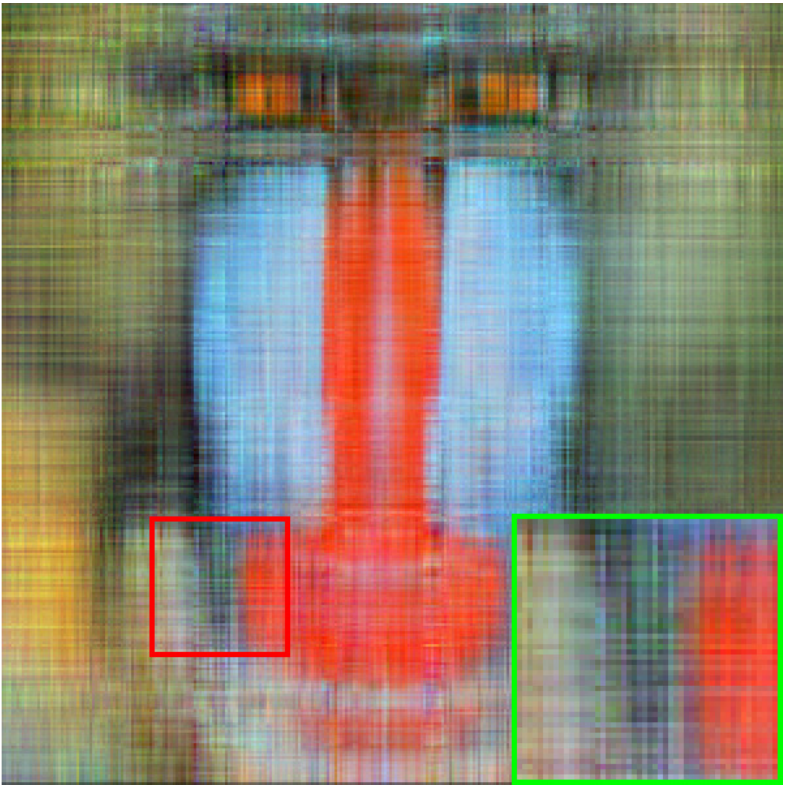}
}
\hspace{-0.5cm}
\subfigure[LRMF]{
	\includegraphics[width=3cm]{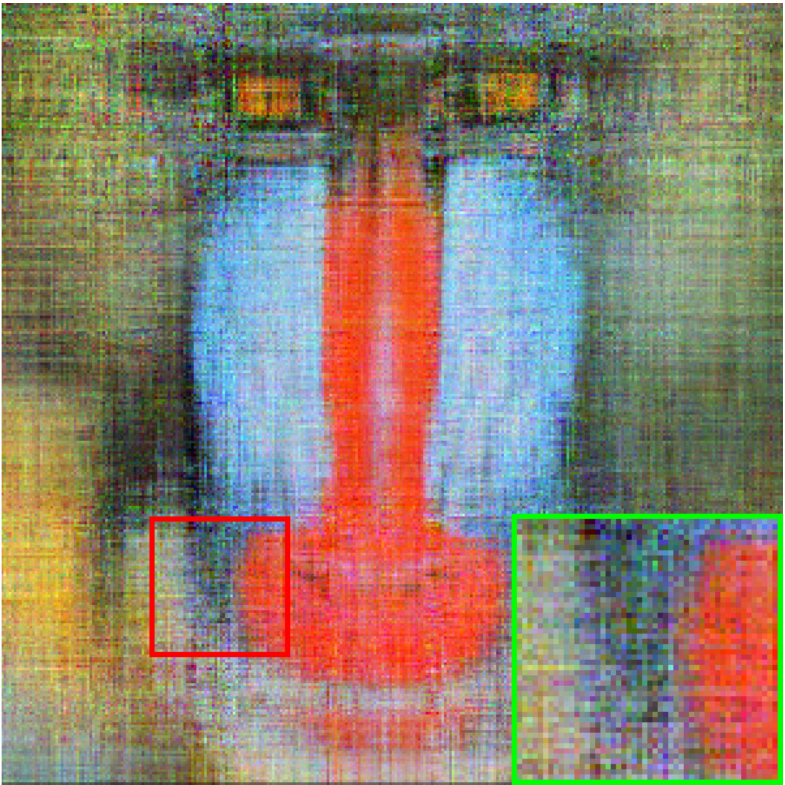}
}
\hspace{-0.5cm}
\subfigure[LRQA-2	]{
	\includegraphics[width=3cm]{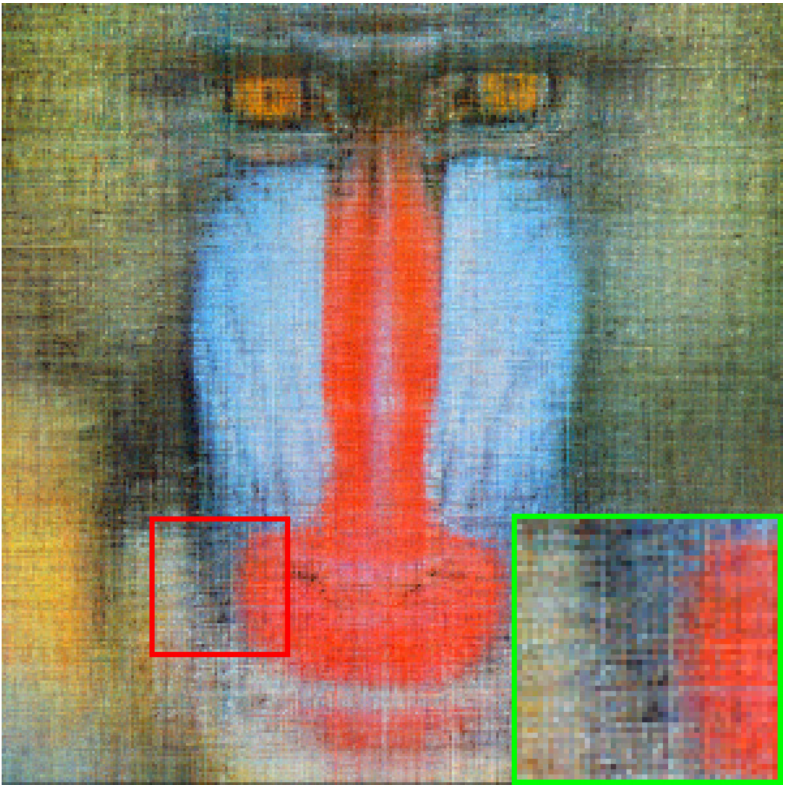}
}
\\
	\centering
\subfigure[	Q-DNN]{
	\includegraphics[width=3cm]{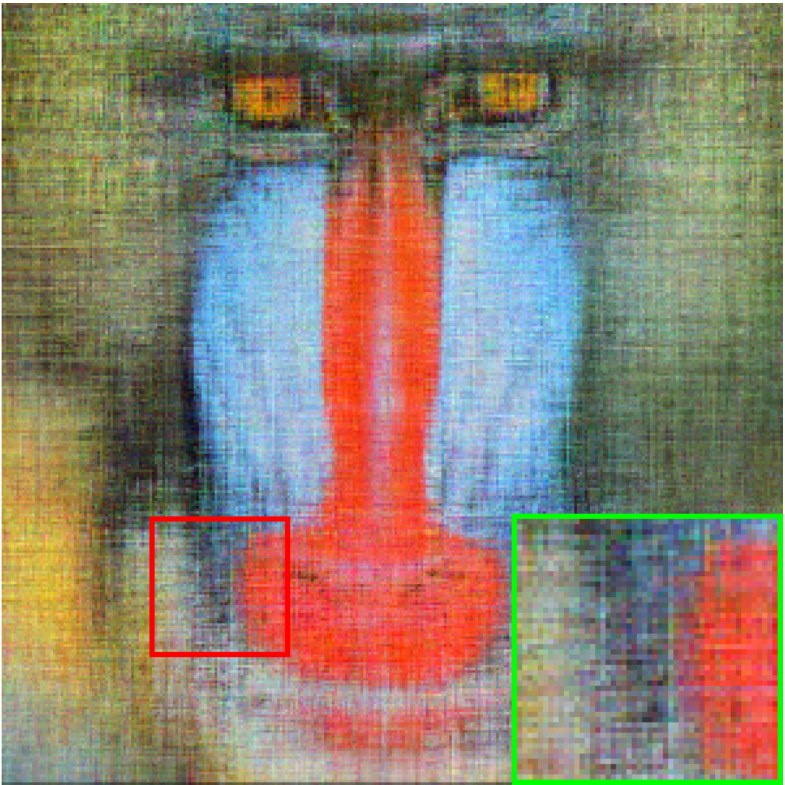}
}
\hspace{-0.5cm}
\subfigure[Q-FNN	]{
	\includegraphics[width=3cm]{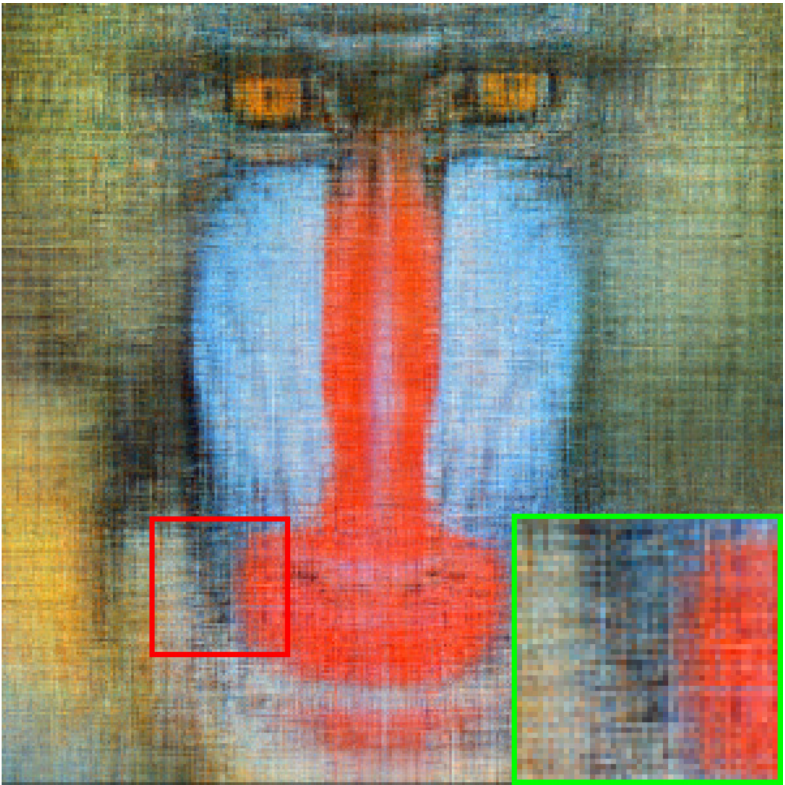}
}
\hspace{-0.5cm}
\subfigure[	QLNF		]{
	\includegraphics[width=3cm]{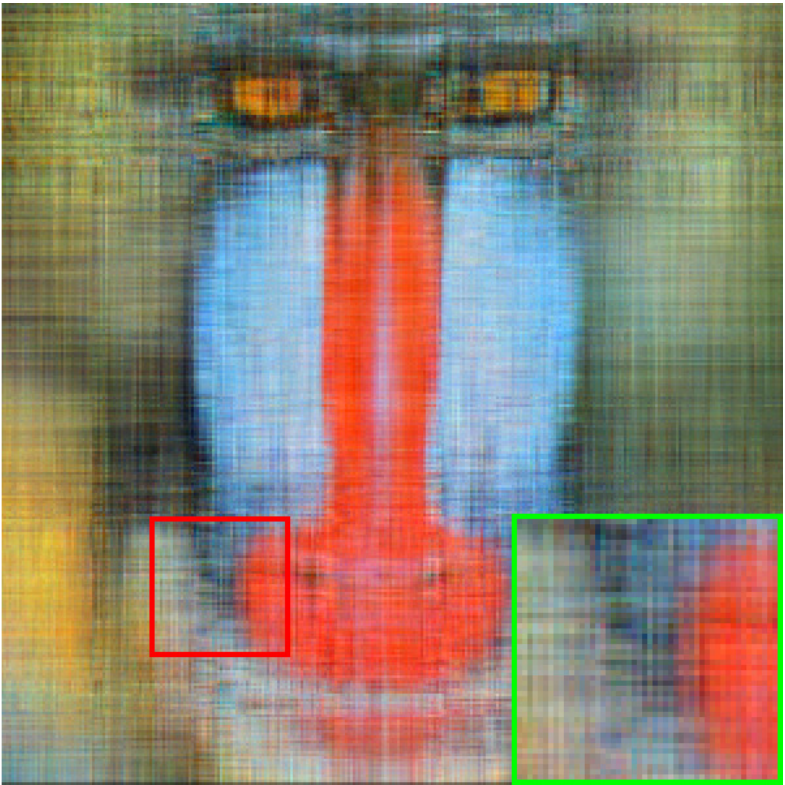}
}
\hspace{-0.5cm}
\subfigure[	TQLNA]{
	\includegraphics[width=3cm]{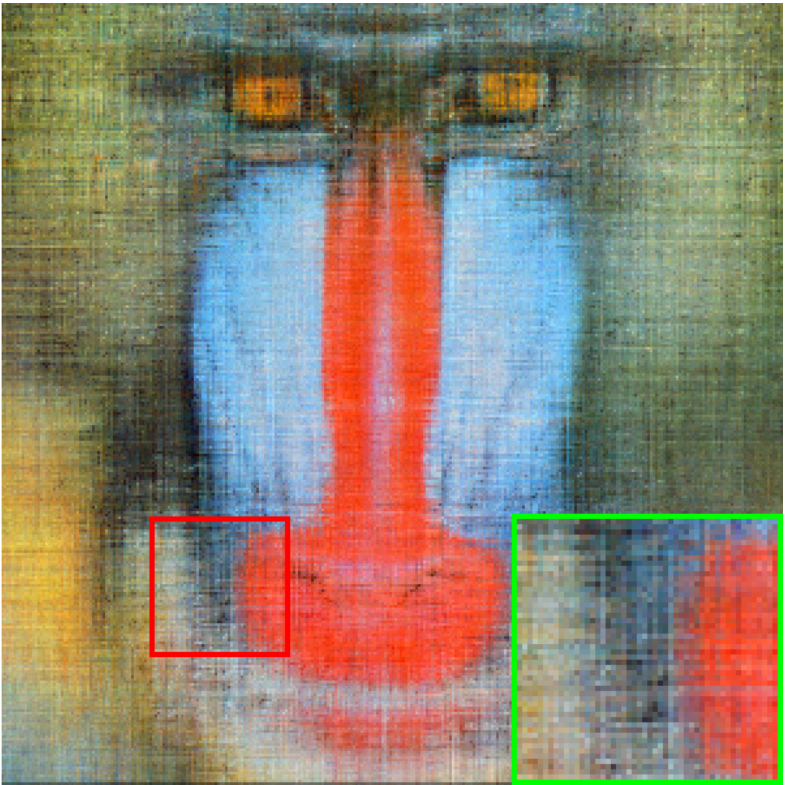}
}
	\caption{  Color image completion results on Image (8). (a) is the original image. (b) is the observed image ( SR=20\%). (c)–(l) are the completion results of  TNNR, WNNR, D-N, F-N, LRMF, LRQA-2, Q-DNN, Q-FNN, QLNF , and TQLNA, respectively. }
	\label{baboon}
\end{figure*}
\begin{table}
	\caption{The corresponding PSNR and SSIM of Figure \ref{baboon}}
	\scalebox{0.7}{\begin{tabular}{ccccccccccc}
			\hline
			Method & TNNR& WNNR& D-N& F-N& LRMF& LRQA-2& Q-DNN& Q-FNN& QLNF &  TQLNA\\
			\hline
			PSNR & 21.175 &18.233 &20.591 &20.682 &19.120 &21.224 &	20.930&21.143 &21.163 &	\textbf{21.593 }
			\\
			\hline
			SSIM & 0.662 &	0.517 &	0.609&	0.612&	0.560 &	0.662 &0.643& 	0.648 &	0.636& 	\textbf{0.680 }
			\\
			\hline
		\end{tabular}\label{babooni}}
\end{table}
\begin{figure}
	\centering
	\includegraphics[width=10cm,keepaspectratio]{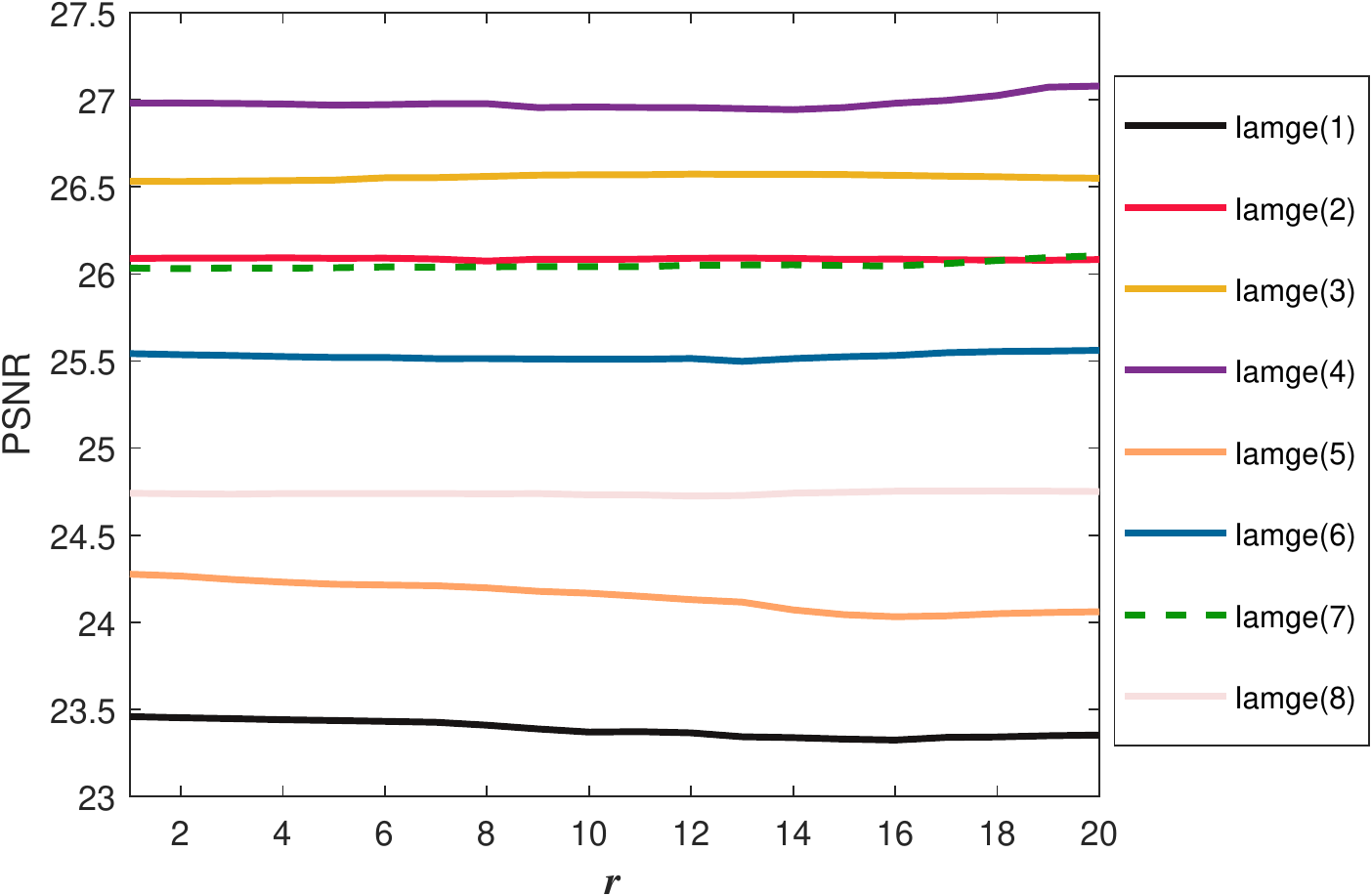}
	\caption{PSNR values of 8 images with different truncated number $r$ when  SR=40\%.}
	\label{r}
\end{figure}
\begin{figure}
	\centering
	\includegraphics[width=14cm,keepaspectratio]{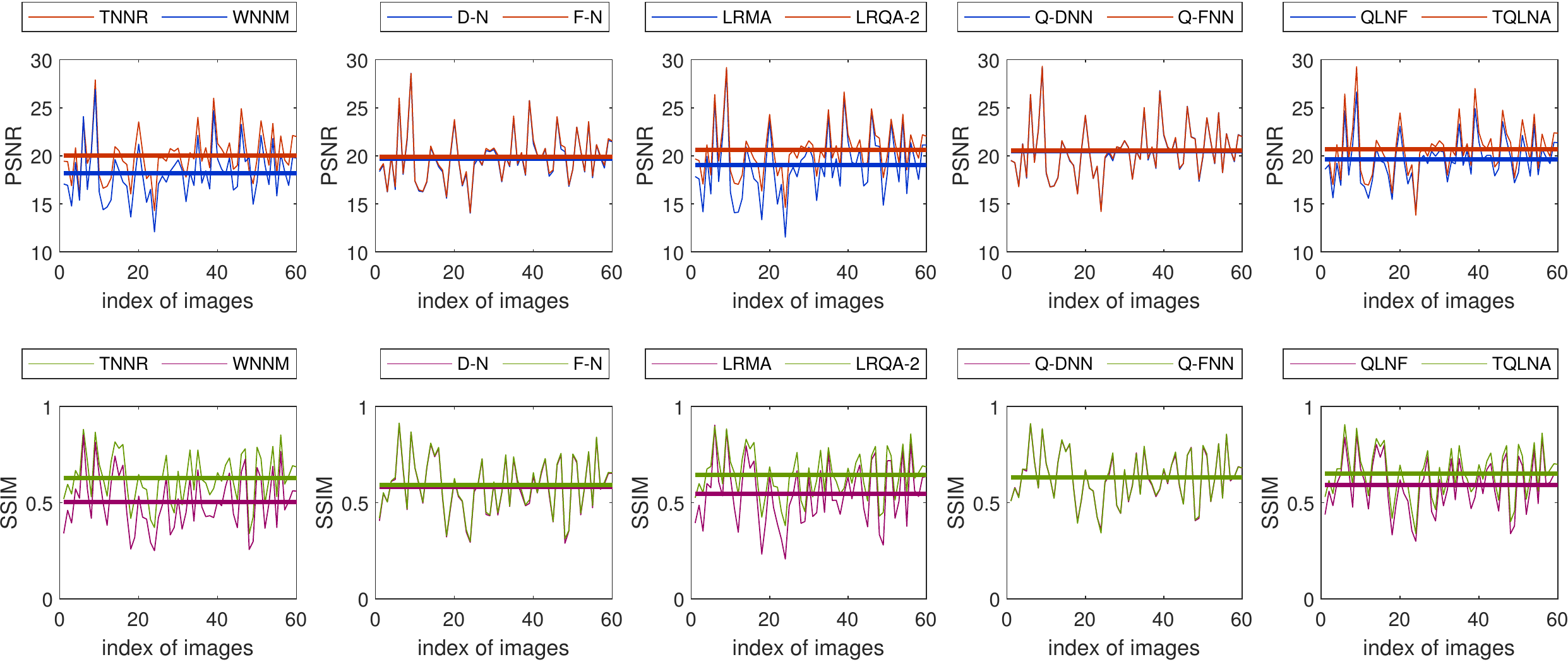}
	\caption{PSNR and SSIM values of 60 images obtained by different compared methods when let SR=20\%. All the straight lines are the average recovered results of these 60 images.}
	\label{0.8}
\end{figure}
\begin{table}
	\caption{The corresponding averge PSNR and SSIM of Figure \ref{0.8}}
	\scalebox{0.7}{\begin{tabular}{ccccccccccc}
			\hline
			Method & TNNR& WNNR& D-N& F-N& LRMF& LRQA-2& Q-DNN& Q-FNN& QLNF &  TQLNA\\
			\hline						
			
			PSNR & 20.043&18.180 &19.718 &19.872	&	19.025&		20.623	&	20.495&	20.535	&19.611&	\textbf{20.703}
			\\
			\hline
			SSIM &0.629 &	0.503 	&	0.583 &	0.591 	&	0.545 	&	0.643 	&	0.629 	&	0.632 &		0.592 	&	\textbf{0.650 }
			\\
			\hline%
			Aver. Time &159.981 &14.000 &1.217 	&\textbf{	0.813} 	&	1.285 	&	59.278 	&29.313 &30.163 &18.685 &22.436 
			\\
			\hline
		\end{tabular}\label{0.8i}}
\end{table}
\begin{figure}
	\centering
	\includegraphics[width=14cm,keepaspectratio]{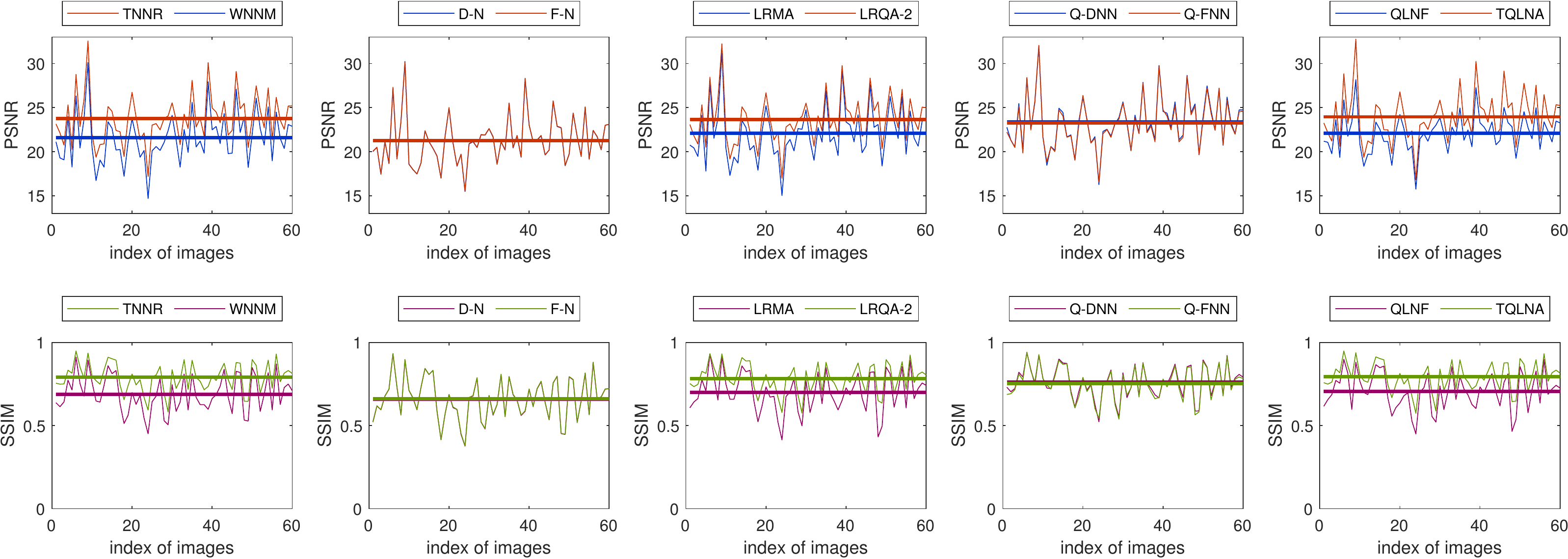}
	\caption{PSNR and SSIM values of 60 images obtained by different compared methods when let  SR=40\%. All the straight lines are the average recovered results of these 60 images.}
	\label{0.6}
\end{figure}
\begin{table}
	\caption{The corresponding averge PSNR and SSIM of Figure \ref{0.8}}
	\scalebox{0.7}{\begin{tabular}{ccccccccccc}
			\hline%
			
			Method & TNNR& WNNR& D-N& F-N& LRMF& LRQA-2& Q-DNN& Q-FNN& QLNF &  TQLNA\\
			\hline						
			
			PSNR &23.775 &	21.606 	&	21.245& 		21.264 	&	22.093 		&23.667 	&	23.388 	&	23.273 	&	22.099 	&	\textbf{23.959 }
			\\
			\hline%
			
			SSIM &0.791 &	0.690 	&	0.660 	&0.661 	&0.700 	&	0.782 &	0.761 &	0.753 	&	0.706 	&	\textbf{0.793 }
			\\
			\hline%
			
			Aver. Time &61.451 	&	30.491 	&	1.405 	&\textbf{	0.981} 	&	0.936 	&	43.223 &		28.437 	&	31.051 	&	56.257 	&	9.789 
			\\
			\hline
		\end{tabular}\label{0.6i}}
\end{table}
\section{Conclusion}
\label{Conclusion}
In this paper, we propose two models for low-rank  quaternion matrix completion by (1) combining low-rank quaternion matrix factorization and logarithmic norm; (2) combining logarithmic norm to depict low rank and truncated strategy. Meanwhile, an efficient FISTA algorithm is developed to solve the proposed QLNF model, and ADMM framework is used to TQLNA model.  Numerical results demonstrate that both visual and quantitative analysis can achieve superior performance by operating our methods.

In order to optimize the timewaster of  QSVD, our future work is to find a more effective way to describe low-rank characteristics  such as taking advantage of some convolutional neural networks (CNN), and further employ it to the quaternion domain \cite{DBLP:conf/eccv/ZhuXXC18}.

\appendix

\section{The PROOF OF Theorem \ref{theorem3}  }
\label{A1}
To prove the theorem \ref{theorem3}, we first introduce the following lemmas.
\begin{lemma}\label{l1}
Let two quaternion matrices $\dot{\mathbf{U}}\in \mathbb{H}^{M\times r}$ and $\dot{\mathbf{V}}\in \mathbb{H}^{N\times r}$ be given with $r\leq\min{M, N}$, for $k=1, \cdots, r$, we have
\begin{equation}\label{am1}
\prod_{i=1}^k\sigma_i(\dot{\mathbf{U}}\dot{\mathbf{V}}^H)\leq\prod_{i=1}^k\sigma_i(\dot{\mathbf{U}})\sigma_i(\dot{\mathbf{V}}).
\end{equation}
\end{lemma}

\begin{proof}
	\begin{equation}
	\begin{aligned}
	\prod_{i=1}^k\sigma_i(\dot{\mathbf{U}}\dot{\mathbf{V}}^H)&=(\prod_{i=1}^k\sigma_i(\dot{\mathbf{U}}\dot{\mathbf{V}}^H)_c)^\frac{1}{2}\\&=(\prod_{i=1}^k\sigma_i(\dot{\mathbf{U}}_c(\dot{\mathbf{V}}^H))_c)^\frac{1}{2}\\&\leq(\prod_{i=1}^k\sigma_i(\dot{\mathbf{U}}_c)\sigma_i(\dot{\mathbf{V}}^H)_c)^\frac{1}{2}\\&=\prod_{i=1}^k\sigma_i(\dot{\mathbf{U}})\prod_{i=1}^k\sigma_i(\dot{\mathbf{V}}),
	\end{aligned}
	\end{equation}
	where the first, second and last equalities can be proved by (\ref{ct})  directly (or refer to \cite{zhang1997quaternions} for more detail), and the inequality follows from \cite{DBLP:books/daglib/0019186} (Theorem 3.3.4)
\end{proof}
Based on lemma 1, we can take the logarithm to both sides of (\ref{am1}), then for $ k=1, \cdots, r$ 
\begin{equation}\label{am2}
\sum_{i=1}^k\log (\sigma_i(\dot{\mathbf{U}}\dot{\mathbf{V}}^H))\leq\sum_{i=1}^k\log(\sigma_i(\dot{\mathbf{U}})\sigma_i(\dot{\mathbf{V}})).
\end{equation}

\begin{lemma}\label{l2}(\cite{DBLP:books/daglib/0019186} ( Corollary 3.2.11))
	Let $\mathbf{x}=[x_i]$ and $\mathbf{y}=[y_i]\in \mathbb{R}^{r}$ be given. $x_1\geqslant \cdots \geqslant x_r\geqslant0$ and $y_1\geqslant \cdots \geqslant y_r\geqslant0$. Then $\sum_{i=1}^kx_i\leq\sum_{i=1}^ky_i$ hold for $k=1, \cdots, r$ if and only if there exists a doubly stochastic matrix $\mathbf{S}\in\mathbf{R}^{r\times r}$ such that the entry-wise inequalities $\mathbf{x}\leq\mathbf{S}\mathbf{y}$ hold.
\end{lemma}
Based on lemma 2, for inequality (\ref{am2}),  when $1\leq k\leq r$, there is a doubly stochastic matrix$\mathbf{S}\in\mathbf{R}^{n\times n}$ such that for $ i=1, \cdots, r$
\begin{equation}\label{am3}
\log (\sigma_i(\dot{\mathbf{U}}\dot{\mathbf{V}}^H))\leq\sum_{i=1}^k\mathbf{S}_{ij}\log(\sigma_j(\dot{\mathbf{U}})\sigma_j(\dot{\mathbf{V}})),
\end{equation}
where $\mathbf{S}_{ij}$ denotes the $(i, j)$-th element of doubly stochastic matrix$\mathbf{S}$ ($\mathbf{S}_{ij} \geqslant 0$, $\sum_{i=1}^r\mathbf{S}_{ij}\geqslant1$, and $\sum_{j=1}^r\mathbf{S}_{ij}\geqslant1$). 
Then we have the following lemma, which follows from \cite{DBLP:books/daglib/0019186} (Theorem 3.3.14)
\begin{lemma}\label{l3}
	Given two quaternion matrices $\dot{\mathbf{U}}\in \mathbb{H}^{M\times r}$ and $\dot{\mathbf{V}}\in \mathbb{H}^{N\times r}$  with $r\leq\min{M, N}$.  Let $\mathit{f}(\cdot)$ be a real-valued function such that $\varphi(x)\equiv\mathit{f}(e^x)$ is convex and increasing on the range $[\min\{\sigma_r(\dot{\mathbf{U}}\dot{\mathbf{V}}^H), \sigma_r(\dot{\mathbf{U}})\sigma_r(\dot{\mathbf{V}})\}, \sigma_1(\dot{\mathbf{U}})\sigma_1(\dot{\mathbf{V}})]$, then for $k=1, \cdots, r$, we have
	\begin{equation}\label{am4}
	\sum_{i=1}^k\mathit{f}(\sigma_i(\dot{\mathbf{U}}\dot{\mathbf{V}}^H))\leq\sum_{i=1}^k\mathit{f}(\sigma_i(\dot{\mathbf{U}})\sigma_i(\dot{\mathbf{V}})).
	\end{equation}
\end{lemma}

\begin{proof}
Because $\varphi(\cdot)$ is an increasing and convex function on the given interval, for inequality (\ref{am3}) we have
\begin{equation}\label{am5}
\begin{aligned}
\sum_{i=1}^k\varphi(\log (\sigma_i(\dot{\mathbf{U}}\dot{\mathbf{V}}^H)))&\leq\sum_{i=1}^k\varphi(\log(\sum_{i=1}^k\mathbf{S}_{ij}\sigma_j(\dot{\mathbf{U}})\sigma_j(\dot{\mathbf{V}})))\\&\leq\sum_{i=1}^k\sum_{j=1}^k\mathbf{S}_{ij}\varphi(\log(\sigma_j(\dot{\mathbf{U}})\sigma_j(\dot{\mathbf{V}})))\\&=\sum_{j=1}^k(\sum_{i=1}^k\mathbf{S}_{ij})\varphi(\log(\sigma_j(\dot{\mathbf{U}})\sigma_j(\dot{\mathbf{V}})))\\&=\sum_{j=1}^k\varphi(\log(\sigma_j(\dot{\mathbf{U}})\sigma_j(\dot{\mathbf{V}}))),
\end{aligned}
\end{equation}
which means that $\sum_{i=1}^k\mathit{f}(\sigma_i(\dot{\mathbf{U}}\dot{\mathbf{V}}^H))\leq\sum_{i=1}^k\mathit{f}(\sigma_i(\dot{\mathbf{U}})\sigma_i(\dot{\mathbf{V}}))$ hold for $k=1, \cdots, r$.
\end{proof}

Proof of Theorem 1:

\begin{proof}
	Since the decomposition  of  $\dot{\mathbf{X}}=\dot{\mathbf{U}}\dot{\mathbf{V}}^H$, where $\dot{\mathbf{U}}\in \mathbb{H}^{M\times r}$ and $\dot{\mathbf{V}}\in \mathbb{H}^{N\times r}$. Denoting $K=\min(M, N, r)$, next we have 
		\begin{equation}
		\begin{aligned}
	2\parallel\dot{\textbf{X}}\parallel_L^{1/2}&=2\sum_{i=1}^{K}\log(\sigma_{i}^{1/2}(\dot{\mathbf{X}})+\epsilon)=\sum_{i=1}^{K}2\log(\sigma_{i}^{1/2}(\dot{\mathbf{U}}\dot{\mathbf{V}}^H)+\epsilon)\\&\leq\sum_{i=1}^{K}2\log(\sigma_{i}^{1/2}(\dot{\mathbf{U}})\sigma_{i}^{1/2}(\dot{\mathbf{V}})+\epsilon)\\&=\sum_{i=1}^{K}\log(\sigma_{i}(\dot{\mathbf{U}})\sigma_{i}(\dot{\mathbf{V}})+2\sqrt{\sigma_{i}(\dot{\mathbf{U}})\sigma_{i}(\dot{\mathbf{V}})}\epsilon+\epsilon^2)\\&\leq\sum_{i=1}^{K}\log(\sigma_{i}(\dot{\mathbf{U}})\sigma_{i}(\dot{\mathbf{V}})+(\sigma_{i}(\dot{\mathbf{U}})+\sigma_{i}(\dot{\mathbf{V}}))\epsilon+\epsilon^2)\\&=\sum_{i=1}^{K}\log(\sigma_{i}((\dot{\mathbf{U}})+\epsilon)(\sigma_{i}(\dot{\mathbf{V}})+\epsilon))\\&=\sum_{i=1}^{K}\log(\sigma_{i}((\dot{\mathbf{U}})+\epsilon))+\sum_{i=1}^{K}\log((\sigma_{i}(\dot{\mathbf{V}})+\epsilon))\\&\leq\sum_{i=1}^{\min(M, r)}\log(\sigma_{i}((\dot{\mathbf{U}})+\epsilon))+\sum_{i=1}^{\min(N, r)}\log((\sigma_{i}(\dot{\mathbf{V}})+\epsilon))\\&=\parallel\dot{\textbf{U}}\parallel_L^{1}+\parallel\dot{\textbf{V}}\parallel_L^{1},
	\end{aligned}
	\end{equation}
where the first inequality follows from Lemma \ref{l3}, as $\varphi(x)=\mathit{f}(e^x)=\log(e^{\frac{x}{2}}+\epsilon)$ which is an increasing and convex function on $[0, +\infty]$. Then the second inequality follows from the Young's inequality \cite{young1912classes}. The last inequality holds due to  $K=\min(M, N, r)\leq\min(M, r)$ and $K=\min(M, N, r)\leq\min(N, r)$ always hold.

Afterwards,  the QSVD of $\dot{\mathbf{X}}$ is $\dot{\mathbf{X}}=\dot{\mathbf{A}}_{\dot{\mathbf{X}}}\dot{\Sigma}_{\dot{\mathbf{X}}}\dot{\mathbf{B}}_{\dot{\mathbf{X}}}^H$. Let $\dot{\mathbf{U}}_\star=\dot{\mathbf{A}}_{\dot{\mathbf{X}}}\dot{\Sigma}_{\dot{\mathbf{X}}}^{\frac{1}{2}}$ and $\dot{\mathbf{V}}_\star=\dot{\mathbf{B}}_{\dot{\mathbf{X}}}\dot{\Sigma}_{\dot{\mathbf{X}}}^{\frac{1}{2}}$. Then we can obtain $\dot{\mathbf{X}}=\dot{\mathbf{U}}_\star\dot{\mathbf{V}}_\star^H$ and $\parallel\dot{\mathbf{X}}\parallel_L^{1/2}=\frac{1}{2}(\parallel\dot{\mathbf{U}}_\star\parallel_L^1+\parallel\dot{\mathbf{V}}_\star\parallel_L^1)$ based on the definition of the logarithmic norm with $p=1/2$.

To sum up, we have 
	\begin{equation}
\parallel\dot{\textbf{X}}\parallel_{L}^{{1}/{2}}=\mathop{\min_{
	\dot{\textbf{U}}, \dot{\textbf{V}}      \atop \dot{\textbf{X}}=\dot{\textbf{U}}\dot{\textbf{V}}^H}} \frac{1}{2}\parallel\dot{\textbf{U}}\parallel_L^1+\frac{1}{2}\parallel\dot{\textbf{V}}\parallel_L^1.
\end{equation}	

\end{proof}

\section{The PROOF OF Theorem \ref{theorem5}}
\label{A2}
\begin{proof}
Assuming that $M\leq N$ the QSVD of the quaternion matrix $\dot{\textbf{X}}\in \mathbb{H}^{M\times N}$ can be represented as:
\begin{equation}
\dot{\mathbf{X}}={\dot{\mathbf{U}}
	\left( \begin{array}{cc}
	\mathbf{\Sigma}_{\tilde{r}} & \mathbf{0}  \\
	\mathbf{0} & \mathbf{0}\\
	\end{array}
	\right )\dot{\mathbf{V}}^H}=\dot{\mathbf{U}}\mathbf{D}\dot{\mathbf{V}}^H
\end{equation}
where $\mathbf{\Sigma}_r=diag({\sigma_1,\cdots, \sigma_{\tilde{r}}})\in\mathbb{R}^{\tilde{r}\times \tilde{r}}$, and all singular values $\sigma_i (i=1,\cdots,r)$ are nonnegative. $\dot{\textbf{U}}\in\mathbb{H}^{M \times M}$ and $\dot{\mathbf{V}}\in\mathbb{H}^{N \times N}$ are two unitary quaternion matrices.
Then,
\begin{equation}
\mid tr(\dot{\mathbf{A}}\dot{\mathbf{X}}\dot{\mathbf{B}}^{H})\mid=
\mid tr(\dot{\mathbf{A}}\dot{\mathbf{U}}\mathbf{D}\dot{\mathbf{V}}^H\dot{\mathbf{B}}^{H})\mid
\end{equation}
Let $\dot{\mathbf{U}}_0=\dot{\mathbf{A}}\dot{\mathbf{U}}=(\dot{u}_{ij})\in\mathbb{H}^{r \times M}$ and $\dot{\mathbf{V}}_0=\dot{\mathbf{B}}\dot{\mathbf{V}}=(\dot{v}_{ij})\in\mathbb{H}^{r \times N}$, distinctly,
$\dot{\mathbf{U}}_0\dot{\mathbf{U}}_0^H=\mathbf{I}_{r\times r}$ and $\dot{\mathbf{V}}_0\dot{\mathbf{V}}_0^H=\mathbf{I}_{r\times r}$. Then we have
\begin{equation}
\begin{aligned}
&\mid tr(\dot{\mathbf{A}}\dot{\mathbf{X}}\dot{\mathbf{B}}^{H})\mid=\mid\dot{\mathbf{U}}_0\mathbf{D}\dot{\mathbf{V}}_0^H\mid
=\mid\sum_{i=1}^r\sum_{j=1}^M\sigma_j\dot{u}_{ij}\bar{\dot{v}}_{ij}\mid
\\&\leq\sum_{i=1}^r\sum_{j=1}^M\sigma_j\mid\dot{u}_{ij}\bar{\dot{v}}_{ij}\mid
=\mid(1,\cdots,1)_{1\times r}\mathbf{P}_{r\times M}(\sigma_i,\cdots,\sigma_M)^T\mid
\\&=\mid(1,\cdots,1,0,\cdots,0)_{1\times M}
\left(
\begin{array}{cc}
\mathbf{P}_{r\times M}  \\
\mathbf{0}_{(M-r)\times M} \\
\end{array}
\right)
(\sigma_i,\cdots,\sigma_M)^T\mid
\end{aligned}
\end{equation}
where $\mathbf{P}_{r\times M}=(\mid\dot{u}_{ij}\bar{\dot{v}}_{ij}\mid)_{r\times M}$.
Because
\begin{equation}
\sum_{i=1}^r\mid\dot{u}_{ij}\bar{\dot{v}}_{ij}\mid\leq\frac{1}{2}[\sum_{i=1}^r\mid\dot{u}_{ij}\mid^2+\sum_{i=1}^r\mid\bar{\dot{v}}_{ij}\mid^2]=1
\end{equation}
\begin{equation}
\sum_{j=1}^M\mid\dot{u}_{ij}\bar{\dot{v}}_{ij}\mid\leq\frac{1}{2}[\sum_{j=1}^M\mid\dot{u}_{ij}\mid^2+\sum_{j=1}^M\mid\bar{\dot{v}}_{ij}\mid^2]\leq1
\end{equation}
, $\mathbf{P}_{r\times M} $ is a doubly-substochastic matrix. According to the theories in \cite{DBLP:books/ap/MarshallO79} (P 136 H.3.b) and \cite{DBLP:books/cu/HJ2012} (Theorem 8.1.4 and Theorem 8.7.6), we have
\begin{equation}
\begin{aligned}
&\mid tr(\dot{\mathbf{A}}\dot{\mathbf{X}}\dot{\mathbf{B}}^{H})\mid
\\&\leq\mid(1,\cdots,1,0,\cdots,0)_{1\times M}
\left(
\begin{array}{cc}
\mathbf{P}_{r\times M}  \\
\mathbf{0}_{(M-r)\times M} \\
\end{array}
\right)
(\sigma_i,\cdots,\sigma_M)^T\mid
\\&\leq\mid(1,\cdots,1,0,\cdots,0)_{1\times M}(\sigma_i,\cdots,\sigma_M)^T\mid
\\&=\sum_{i=1}^r\sigma_i.
\end{aligned}
\end{equation}
When $\dot{\mathbf{A}}= [\mathbf{I}_{r\times r},\mathbf{0}_{(r)\times (M-r)} ]\dot{\mathbf{U}}^H$ and
$\dot{\mathbf{B}}= [\mathbf{I}_{r\times r},\mathbf{0}_{(r)\times (N-r)} ]\dot{\mathbf{V}}^H$, we get
\begin{equation}
\max |tr(\dot{\mathbf{A}}\dot{\mathbf{X}}\dot{\mathbf{B}}^H)|=\sum_{i=1}^r\sigma_{i}(\dot{\mathbf{X}}).
\end{equation}
\end{proof}

\section*{Acknowledgements}
This work was supported by the Science and Technology Development Fund, Macau SAR (File no. FDCT/085/2018/A2) and University of Macau (File no. MYRG2019-00039-FST).





%
%
%
\bibliographystyle{unsrt}
\bibliography{mybibfile}
\end{document}